\documentclass[conference]{IEEEtran}

\usepackage{ifluatex}
\ifluatex\else\pdfoutput=1\fi

\newif\ifarXiv
\arXivfalse %
\arXivtrue  %

\usepackage{macros}

\title{Normalization for Cubical Type Theory}

\author{%
  \IEEEauthorblockN{Jonathan Sterling}%
  \IEEEauthorblockA{Computer Science Department\\ Carnegie Mellon University}
  \and
  \IEEEauthorblockN{Carlo Angiuli}%
  \IEEEauthorblockA{Computer Science Department\\ Carnegie Mellon University}%
}

\IEEEoverridecommandlockouts
\IEEEpubid{\makebox[\columnwidth]{978-1-6654-4895-6/21/\$31.00~
\copyright2021 IEEE \hfill} \hspace{\columnsep}\makebox[\columnwidth]{ }}

\begin{document}
\maketitle

\ifarXiv
\pagestyle{plain}
\fi

\begin{abstract}
  We prove normalization for (univalent, Cartesian) cubical type theory, closing
  the last major open problem in the syntactic metatheory of cubical type
  theory. Our normalization result is \emph{reduction-free}, in the sense of
  yielding a bijection between equivalence classes of terms in context and a
  tractable language of $\beta/\eta$-normal forms. As corollaries we obtain both
  decidability of judgmental equality and the injectivity of type constructors.
\end{abstract}

\section{Introduction}

De Morgan \citep{cchm:2017} and Cartesian
\citep{abcfhl:2019,angiuli-favonia-harper:2018} cubical type theory are recent
extensions of Martin-L\"of type theory which provide constructive formulations
of higher inductive types and Voevodsky's univalence axiom; unlike homotopy type
theory \citep{hottbook}, both enjoy \emph{canonicity}, the property that closed
terms of base type are judgmentally equal to constructors
\citep{huber:2018,angiuli-favonia-harper:2018}.

Several proof assistants already implement cubical type theory, most notably
Cubical Agda \citep{vezzosi-mortberg-abel:2019} (for the De Morgan variant) and
\redtt~\citep{redtt:2018} (for Cartesian).
Like most type-theoretic proof assistants, both typecheck terms using
algorithms inspired by \emph{normalization by evaluation}~\citep{abel:2013},
which interleave evaluation and decomposition of types. The correctness of these
algorithms hinges not on canonicity but on \emph{normalization} theorems
characterizing judgmental equivalence classes of \emph{open} terms---and
consequences thereof, such as decidability of equality and injectivity of type
constructors. But unlike canonicity, normalization and its corollaries have
until now remained conjectures for cubical type theory.

We contribute the first normalization proof for Cartesian cubical type theory
\citep{abcfhl:2019}. By relying on recent advances in the metatheory of type
theory, our proof is significantly more abstract and concise than existing
canonicity proofs for cubical type theory; moreover, it can be adapted to De
Morgan cubical type theory without conceptual changes.

\subsection{Cubical type theory and synthetic semantics}

Cubical type theory extends type theory with a number of features centered
around a primitive interval $\II$ with elements $0,1:\II$. Propositional
equality is captured by a path type $\Con{path}\prn{A,a_0,a_1}$ whose elements
are functions $f:\Prod{i:\II}{A\prn{i}}$ satisfying $f\prn{0}=a_0$ and
$f\prn{1}=a_1$ judgmentally. Congruence of paths follows from substitution; the
remaining properties of equality are defined at each type by the \emph{Kan
operations} of coercion and box filling (or composition).

In Cartesian cubical type theory, coercion is a function
\[
  \Con{coe} : \Prod{A:\II\to\mathcal{U}}{\Prod{r,s:\II}{A\prn{r} \to A\prn{s}}}
\]
satisfying $\Con{coe}\prn{A,r,r,a}=a$ for all $A:\II\to\mathcal{U}$ and
$a:A\prn{r}$,%
  \footnote{In De Morgan cubical type theory, coercion is limited to
  $A\prn{0}\to A\prn{1}$, a restriction counterbalanced by the additional (De
  Morgan) structure on $\II$.}
and additional equations for each particular connective, e.g.:
\begin{align*}
  & \Con{coe}\prn{\lambda i.\prn{\Prod{x:A\prn{i}}{B\prn{i,x}}},r,s,f} \\
  {}={}&
  \lambda x.
  \Con{coe}\prn{
    \lambda i. B\prn{i,\Con{coe}\prn{A,s,i,x}},r,s,
    f\prn{\Con{coe}\prn{A,s,r,x}}
  }
\end{align*}

The equations governing coercion and composition are complex, especially for the
$\Con{glue}$ type which justifies univalence; calculating (and verifying the
well-definedness of) these equations was a major obstacle in early work on
cubical type theory. Orton and Pitts \citep{orton-pitts:2016,lops:2018}
streamlined this process by observing that the model construction for De Morgan
cubical type theory---the most technical part of which is these equations---can
be carried out \emph{synthetically} in the internal extensional type theory of
any topos satisfying nine axioms (e.g., whose objects are De Morgan cubical
sets); \citet{abcfhl:2019} establish an analogous result for the Cartesian
variant.

A subtle aspect of these models is that coercion is defined not on types
$A:\mathcal{U}$ but on type families $A:\II\to\mathcal{U}$; consequently, a
semantic universe $\mathscr{U}$ of types-with-coercion must strangely admit a
coercion structure for every figure $\II\to\mathscr{U}$. \citet{lops:2018}
obtain $\mathscr{U}$ by transposing the coercion map across the right adjoint
to exponentiation by $\II$ given by the tininess of $\II$.

The synthetic approach of Orton and Pitts simplifies and clarifies the model
construction of cubical type theory by factoring out naturality obligations, in
much the same way that homotopy type theory provides a ``synthetic homotopy
theory'' that factors out e.g., continuity obligations.

In this paper, we combine ideas of Orton and Pitts with the \emph{synthetic
Tait computability} (STC) theory of \citet{sterling-harper:2020}, which factors
out bureaucratic aspects of syntactic metatheory. In STC, one considers an
extensional type theory whose types are (proof-relevant) logical relations; the
underlying syntax is exposed via a proof-irrelevant proposition $\Syn$ under
which the syntactic part of a logical relation is projected.

\subsection{Canonicity for cubical type theory}

Traditional canonicity proofs fix an evaluation strategy for closed terms, and
associate to each closed type a proof-irrelevant \emph{computability predicate}
or \emph{logical relation} ranging over closed terms of that type. Then, one
ensures that evaluation is contained within judgmental equality, that
computability is closed under evaluation, and that computability at base type
implies evaluation to a constructor; canonicity follows by proving that every
well-typed closed term is computable.

Whereas evaluation in ordinary type theory need not descend under binders,
evaluating a closed coercion $\Con{coe}\prn{\lambda i.A,r,s,a}$ in cubical type
theory requires determining the head constructor of (i.e., evaluating) the type
$A$ in context $i:\II$. Accordingly, the cubical canonicity proofs of
\citet{huber:2018,angiuli-harper-wilson:2017,angiuli-favonia-harper:2018}
define evaluation and computability for terms in context
$\prn{i_1:\II,\dots,i_n:\II}$. In both proofs, the difficulty arises that
typing and thus computability are closed under substitutions of the form
$\Mor{\II^n}{\II^m}$, but evaluation is not; both resolve the issue by showing
evaluation is closed under computability up to judgmental equality.

\subsection{Semantic and proof-relevant computability}

The past several years have witnessed an explosion in \emph{semantic}
computability techniques for establishing syntactic metatheorems
\citep{altenkirch-hofmann-streicher:1995,fiore:2002,%
shulman:2015,uemura:2017,coquand:2019,%
kaposi-huber-sattler:2019,coquand-huber-sattler:2019,sterling-angiuli-gratzer:2019,sterling-harper:2020}.
What makes semantic computability different from ``free-hand'' computability is
that it is expressed as a \emph{gluing model}, parameterized in the generic model of
the type theory; hence one is always working with typed terms up to judgmental
equality.

A new feature of semantic computability, forced in many cases by the absence of
raw terms, is that a term may be computable in more than one way. This
\emph{proof-relevance} plays an important role in the normalization arguments
of \citet{altenkirch-hofmann-streicher:1995,fiore:2002,coquand:2019} as well as
the canonicity arguments of \citet{coquand:2019,sterling-angiuli-gratzer:2019}.
The proof-relevant approach is significantly simpler to set up than the
alternative and it provides a compositional account of computability for
universe hierarchies, which had been the main difficulty in conventional
free-hand arguments.

\begin{example}
  A semantic canonicity argument for ordinary type theory associates to each
  closed type a \emph{computability structure} which assigns to each equivalence
  class of closed terms of that type a set of ``computability proofs.'' We choose
  the computability structure at base type to be a collection of ``codes'' for each
  constant; then, by exhibiting a choice of computability proof for every
  well-typed term, we conclude that every equivalence class at base type is a
  constant.
\end{example}

Crucially, the ability to store data within the computability proofs
circumvents the need to define a subequational evaluation function, allowing us
to carry out the entire argument over equivalence classes of terms; rather than
choosing a representative of each equivalence class, we encode canonical forms
as a \emph{structure} indexed over equivalence classes of terms.

\subsubsection{In what contexts do we compute?}

Semantic computability arguments have already been used to establish
ordinary canonicity, cubical canonicity, and ordinary normalization. The key
difference between these arguments lies in what is considered an ``element'' of
a type, or more precisely, what are the contexts (and substitutions) of
interest. In ordinary canonicity proofs, the only context of interest is the
closed context, in which each type has just a \emph{set} of elements; the
computability structures are thus families of sets.

Following Huber and Angiuli, Favonia, and Harper
\citep{huber:2018,angiuli-favonia-harper:2018}, cubical canonicity proofs must
consider terms in all contexts $\II^n$, with all substitutions
$\Mor{\II^n}{\II^m}$ between them. These contexts and substitutions induce a
\emph{cubical set} (i.e., a presheaf) of elements of each type; the
computability structures in question are thus families of cubical sets indexed in
the application of a ``cubical nerve'' applied to a syntactical object, an
arrangement suggested by Awodey in 2015. Notably, because semantic
computability arguments are not evaluation-based, cubical canonicity proofs in
this style (e.g., that of Sterling, Angiuli, and Gratzer
\citep{sterling-angiuli-gratzer:2019}) entirely sidestep the evaluation
coherence difficulties of prior work
\citep{huber:2018,angiuli-favonia-harper:2018}.

The passage to presheaves of elements is not a novel feature of cubical
canonicity; it appears already in normalization proofs for ordinary type theory,
in the guise of Kripke logical relations of varying arities
\citep{jung-tiuryn:1993}. Because normalization is by definition a
characterization of open terms, one must necessarily consider the presheaf of
elements of a type relative to \emph{all} contexts; but in light of the fact
that normal forms are not closed under substitutions of terms for variables, one
considers only a restricted class of substitutions (e.g., weakenings, injective
renamings, or all variable renamings).

Following Tait \citep{tait:1967}, the normal forms of type theory are defined
mutually as the \emph{neutral forms} $\Ne{}{A}$ (variables and eliminations thereof)
and the \emph{normal forms} $\Nf{A}$ (constants and constructors applied to
normals) of each open type $A$. Proof-irrelevant normalization arguments then
establish that every neutral term is computable (via \emph{reflection}
$\Reflect{A}{}$), and that every computable term has a normal form (via
\emph{reification} $\Reify{A}$). Proof-relevant normalization arguments follow
the same yoga of reflection and reification, except that we speak not of the
subset of neutral terms but rather the collection of neutrals and normals
encoding each equivalence class of terms \citep{coquand:2019}.

\subsubsection{Related work on gluing for type theory}

The past forty years have brought a steady stream of research developing the
\emph{gluing} perspective on logical
relations~\citep{freyd:1978,crole:1993,altenkirch-hofmann-streicher:1995,streicher:1998,fiore:2002};
however, only in the past several years has our understanding of the syntax and
semantics of dependent
types~\citep{awodey:2018:natural-models,newstead:2018,altenkirch-kaposi:2016,uemura:2019,gratzer-sterling:2020}
caught up with the mathematical tools required to advance a truly objective
metatheory of dependent type theory.

In particular, Coquand's analysis \citep{coquand:2019} of proof-relevant
canonicity and normalization arguments for dependent type theory in terms of
categorical gluing was the catalyst for a number of recent works that obtain
non-trivial metatheorems for dependent type theory by semantic means, although
some years earlier \citet{shulman:2015} had already used gluing to prove
a homotopy canonicity result for a univalent type theory.

\citet{uemura:2017} proved a general gluing theorem for certain dependent
type theories in the language of Shulman's type theoretic fibration categories;
\citet{kaposi-huber-sattler:2019} proved a similar
result in the language of categories with families.
\citet{coquand-huber-sattler:2019} employed gluing to
prove a homotopy canonicity result for a version of cubical type theory that
omits certain computation rules,
and \citet{kapulkin-sattler:2019} used gluing to prove
homotopy canonicity for homotopy type theory (as famously conjectured by
Voevodsky).
\citet{sterling-angiuli-gratzer:2019} adapted
Coquand's gluing argument to prove the first non-operational strict canonicity
result for a cubical type theory.
\citet{gratzer-kavvos-nuyts-birkedal:2020} used gluing to prove
canonicity for a general \emph{multi-modal} dependent type
theory.
\citet{sterling-harper:2020} employ a different gluing
argument to establish a proof-relevant generalization of the Reynolds
Abstraction Theorem for a calculus of ML modules.
\citet{gratzer-sterling:2020} used gluing to establish the conservativity of higher-order
judgments for dependent type theories.

\subsubsection{What are the neutrals of cubical type theory?}\label{sec:what-are-the-neutrals}

Today's obstacles to proving cubical normalization are entirely different from
the obstacles faced in the first proofs of cubical canonicity
\citep{huber:2018,angiuli-favonia-harper:2018}. As we have already discussed,
coherence of evaluation is a non-issue for semantic computability;
moreover, as normalization already descends under binders, this feature of
coercion poses no additional difficulty.

However, cubical type theory includes a number of open judgmental equalities
that challenge the yoga of reflection and reification. Consider the rule that
applying any element, even a variable, of type $\Con{path}\prn{A,a_0,a_1}$ to
$0:\II$ (resp., $1:\II$) equals $a_0$ (resp., $a_1$). Whereas application of
neutrals (e.g., variables) to normal forms (e.g., constants) is typically
irreducible, here path application of a variable to a constant (but \emph{not}
to a variable) is a redex which may uncover further redexes:
\[
  x:\Con{path}\prn{\lambda\_.\mathbb{N}\to\mathbb{N},\Con{fib},\Con{fib}} \vdash
  x\prn{0}\prn{7} = 13 : \mathbb{N}
\]

One might imagine defining the normal form of path application by a case split
on the elements of $\II$ (sending $0,1:\II$ to the normal form of
$\Con{fib}\prn{7}$, and $i:\II$ to a neutral application), but such a case split
requires us to model $\II$ as a coproduct, which will not be tiny, preventing us
from obtaining a universe of Kan types following
\citet{lops:2018}.

Similar issues arise with a number of equations in both Cartesian and De Morgan
cubical type theory; in fact, the Cartesian variant is \emph{a priori} more
challenging in this regard because contraction of interval variables $i,j:\II$
may also induce computation (e.g., in $\Con{coe}\prn{A,i,j,a}$).

In this paper, we index neutrals $\Ne{\phi}{A}$ by a proposition $\phi$
representing their \emph{locus of instability}, or where they cease to be
neutral. For example, path application sends a $\phi$-unstable neutral of type
$\Con{path}\prn{A,a_0,a_1}$ and an element $r:\II$ to a
$\prn{\phi\lor\prn{r=0}\lor\prn{r=1}}$-unstable neutral of type $A\prn{r}$.
Reflection $\Reflect{A}{\phi}$ operates on \emph{stabilized neutrals}, pairs of
a neutral $\Ne{\phi}{A}$ with a proof that the term is computable under $\phi$.
In the case of path application $x\prn{r}$, one must provide computability
proofs for $a_0,a_1$ under the assumption $r=0\lor r=1$.

Terms in the ordinary fragment are never unstable (hence $\phi=\bot$), in which
case a stabilized neutral is a neutral in the ordinary sense; ``neutrals'' with
cubical redexes (such as $x\prn{0}$) have $\phi=\top$, in which case their
stabilized neutral is just a computability proof (and $\Reflect{A}{\top}$ is the
identity).  To our knowledge, this is the first time that computability data
appears in the domain of the reflection operation.

\subsection{Contributions}

We establish the normalization theorem (\cref{thm:soundness}) for Cartesian
cubical type theory closed under $\Pi$, $\Sigma$, path, glue, and a higher
inductive circle type, using a cubical extension of synthetic Tait
computability \citep{sterling-harper:2020}; the new idea on which our argument
hinges is the concept of \emph{stabilized neutrals} described above. As
corollaries to our main result, we obtain the admissible injectivity of type
constructors (\cref{thm:injectivity}) as well as an algorithm to decide
judgmental equality (\cref{thm:dec-eq}).

The present paper does not describe universes or the modifications necessary to
prove normalization for De Morgan cubical type theory; but note that univalence
can be stated without universes, as we have done here. The novel aspects of
cumulative, univalent universes are already supported because of the tininess
of the interval and the account of glue types; the difference is that the
operator projecting a normal type from a normalization structure of size
$\alpha$ must be generalized over $\beta\geq\alpha$. Our argument carries over
\emph{mutatis mutandis} to a normalization proof for De Morgan cubical type
theory.

In \cref{sec:syntax} we discuss the syntax of Cartesian cubical type theory and
its situation within a dependently sorted logical framework. In \cref{sec:stc},
we axiomatize a cubical version of synthetic Tait computability
(STC)~\citep{sterling-harper:2020}, a modal type theory of synthetic logical
relations suitable for proving syntactic metatheorems; we construct in cubical
STC a ``normalization model'' of cubical type theory displayed over the generic
model. In \cref{sec:computability-topos} we construct a topos model of cubical
STC, which takes us the remaining distance to the main results of this paper,
which are described in \cref{sec:normalization-result}.
\ifarXiv
\else
  For space reasons, some details are available in the appendix
  \citep{sterling-angiuli:2021:extended}.
\fi

\section{Cartesian cubical type theory}\label{sec:syntax}

We define the subject of our normalization theorem, intensional Cartesian
cubical type theory, as a locally Cartesian closed category of judgments
$\ThCat$ generated by the signature in
\cref{fig:cubical-tt:judgments,fig:cubical-tt:connectives}. Readers may consult
\citep{abcfhl:2019} and \citep[Appendix B]{angiuli:2019} for further exposition,
including rule-based presentations.

\begin{figure*}
  \begin{gather*}
    \begin{aligned}
      \II,\FF,\Tp &: \Jdg\\
      \brk{-} &: \IHom{\FF}{\Jdg}\\
      \Tm &: \IHom{\Tp}{\Jdg}\\
      0,1 &: \II\\
      \prn{=} &: \IHom{\prn{\II\times\II}}{\FF}\\
      \prn{\land_\FF},\prn{\lor_\FF} &: \IHom{\prn{\FF\times\FF}}{\FF}\\
      \prn{\forall_\II} &: \IHom{\prn{\IHom{\II}{\FF}}}{\FF}
    \end{aligned}
    \qquad\qquad
    \begin{aligned}
      \_ &: \Params{\phi} {\Prod{p,q:\brk{\phi}}{p =\Sub{\brk{\phi}} q}}\\
      \_ &: \Params{\phi,\psi}{\prn{\brk{\phi}\cong\brk{\psi}}\cong\prn{\phi =_{\FF} \psi}}\\
      \_ &: \Params{\phi}\prn{\Prod{i:\II}{\brk{\phi\prn{i}}}} \cong \brk{\forall_\II \phi}\\
      \_ &: \Params{r,s}\prn{r=_\II s}\cong \brk{r=s} \\
      \_ &: \Params{\phi,\psi}\prn{\brk{\phi}\times\brk{\psi}} \cong \brk{\phi \land_\FF \psi} \\
      \_ &: \Params{\phi_0,\phi_1}\IHom{\brk{\phi_i}}{\brk{\phi_0\lor_\FF\phi_1}}
    \end{aligned}
    \\[8pt]
    \begin{aligned}
      \multispan{2}{\underline{\textbf{For each judgment $\Kwd{J}\in\brc{\II,\FF,\Tp,\brk{\phi},\Tm\prn{A}}$:}\hfill}}\\
      \Con{abort}\Sub{\Kwd{J}} &: \IHom{\brk{0=1}}{\ObjTerm{}\cong\Kwd{J}}\\
      \Con{split}\Sub{\Kwd{J}} &:
      \Params{\phi,\psi}
      \Prod{x_\phi : \IHom{\brk{\phi}}{\Kwd{J}}}
      \Prod{x_\psi : \IHom{\brk{\psi}}{\Ext{\Kwd{J}}{\phi}{x_\phi}}}
      \IHom{\brk{\phi\lor_\FF\psi}}{\Kwd{J}}\\
      \_ &: \Params{\phi,\psi}\Prod{\_:\brk{\phi}}{\Con{split}\Sub{\Kwd{J}}\prn{x_\phi,x_\psi} =\Sub{\Kwd{J}} x_\phi}\\
      \_ &: \Params{\phi,\psi}\Prod{\_:\brk{\psi}}{\Con{split}\Sub{\Kwd{J}}\prn{x_\phi,x_\psi} =\Sub{\Kwd{J}} x_\psi}\\
      \_ &: \Params{\phi,\psi,x} x =\Sub{\Kwd{J}} \Con{split}\Sub{\Kwd{J}}\prn{x,x}
    \end{aligned}
    \qquad
    \begin{aligned}
      &\\
      \Ext{\Kwd{J}}{\phi}{x_\phi} &:= \Sum{x:\Kwd{J}}{\Prod{p:\brk{\phi}}{x=\Sub{\Kwd{J}}x_\phi\prn{p}}}\\
      \FakeFalse &:= \prn{0=1}\\
      \partial{i} &:= \prn{i=0}\lor_\FF\prn{i=1}\\
      \brk{}\Sub{\Kwd{J}} &:= \Con{abort}\Sub{\Kwd{J}}\\
      \brk{\phi\hookrightarrow x_\phi, \psi\hookrightarrow x_\psi}\Sub{\Kwd{J}} &:= \Con{split}\Sub{\Kwd{J}}\prn{x_\phi,x_\psi}
    \end{aligned}
  \end{gather*}
  \caption{Basic judgmental structure of Cartesian cubical type theory.}
  \label{fig:cubical-tt:judgments}
\end{figure*}

\begin{figure*}
  \begin{align*}
    \Con{path} &: \IHom{\prn{\Sum{A:\IHom{\II}{\Tp}}{\Tm\prn{A\prn{0}}\times\Tm\prn{A\prn{1}}}}}{\Tp}\\
    \Pi,\Sigma &: \IHom{\prn{\Sum{A:\Tp}{\IHom*{\Tm\prn{A}}{\Tp}}}}{\Tp}\\
    \Con{glue} &:
    \Prod{\phi:\FF}
    \Ext{
      \IHom{
        \prn{
          \Sum{B:\Tp}
          \Sum{A : \IHom{\brk{\phi}}{\Tp}}
          \Prod{\_:\brk{\phi}}{
            \Tm\prn{\Con{Equiv}\prn{A,B}}
          }
        }
      }{
        \Tp
      }
    }{\phi}{
      \lambda \prn{B,A,f}. A
    }
    \\
    \Con{S1} &: \Tp\\
    \Con{path}/\Con{tm} &:
      \Params{A,a_0,a_1}
      \prn{
        \Prod{i:\II}{
          \Ext{\Tm\prn{A\prn{i}}}{\partial{i}}{
            \brk{\overline{i=\epsilon\hookrightarrow a_\epsilon}}\Sub{\Tm\prn{A\prn{i}}}
          }
        }
      }
      \cong
      \Tm\prn{\Con{path}\prn{A,a_0,a_1}}
    \\
    \Pi/\Con{tm} &:
      \Params{A,B}
      \prn{\Prod{x:\Tm\prn{A}}{\Tm{\prn{B\prn{x}}}}} \cong
      \Tm\prn{\Pi\prn{A,B}}
    \\
    \Sigma/\Con{tm} &:
      \Params{A,B}
      \prn{\Sum{x:\Tm\prn{A}}{\Tm{\prn{B\prn{x}}}}} \cong
      \Tm\prn{\Sigma\prn{A,B}}
    \\
    \Con{glue}/\Con{tm} &:
      \Params{\phi,B,A,f}
      \Ext{
        \prn{
          \Sum{a : \Prod{\_:\brk{\phi}}{\Tm\prn{A}}}{
            \Ext{\Tm\prn{B}}{\phi}{f\prn{a}}
          }
        }
        \cong
        \Tm\prn{\Con{glue}\prn{\phi,B,A,f}}
      }{\phi}{
        \lambda\prn{a,b}.a
      }
    \\
    \Con{base} &: \Tm\prn{\Con{S1}}\\
    \Con{loop} &: \Prod{i:\II}{\Ext{\Tm\prn{\Con{S1}}}{\partial{i}}{\Con{base}}}\\
    \Con{ind}\Sub{\Con{S1}} &:
      \Prod{C : \IHom{\Tm\prn{\Con{S1}}}{\Tp}}
      \Prod{b : \Tm\prn{C\prn{\Con{base}}}}
      \Prod{
        l : \Prod{i:\II}{
          \Ext{
            \Tm\prn{C\prn{\Con{loop}\prn{i}}}
          }{\partial{i}}{b}
        }
      }
      \Prod{x : \Tm\prn{\Con{S1}}}
      \Tm\prn{C\prn{x}} \\
    \_ &:
      \Params{C,b,l}
      \Con{ind}\Sub{\Con{S1}}\prn{C,b,l,\Con{base}}
      =\Sub{\Tm\prn{C\prn{\Con{base}}}} b\\
    \_ &:
      \Params{C,b,l,i}
      \Con{ind}\Sub{\Con{S1}}\prn{C,b,l,\Con{loop}\prn{i}}
      =\Sub{\Tm\prn{C\prn{\Con{loop}\prn{i}}}} l\prn{i} \\
    \Con{hcom} &:
      \Prod{A:\Tp}
      \Prod{r,s:\II}
      \Prod{\phi:\FF}
      \Prod{
        a :
        \Prod{i:\II}{
          \Prod{\_:\brk{i=r\lor_\FF\phi}}{\Tm\prn{A}}
        }
      }
      \Ext{\Tm\prn{A}}{r=s\lor_\FF\phi}{a\prn{s}}\\
    \Con{coe} &:
      \Prod{A:\IHom{\II}{\Tp}}
      \Prod{r,s:\II}
      \Prod{a:\Tm\prn{A\prn{r}}}
      \Ext{\Tm\prn{A\prn{s}}}{r=s}{a}
  \end{align*}

  \begin{align*}
    \Con{isContr} &: \IHom{\Tp}{\Tp} \\
    \Con{isContr} &:=
      \lambda A.\Sigma\prn{A,\lambda x.\Pi\prn{A,\lambda y.\Con{path}\prn{\lambda\_.A,x,y}}}
    \\[6pt]
    \Con{Equiv} &: \IHom{\Tp}{\IHom{\Tp}{\Tp}} \\
    \Con{Equiv} &:=
      \lambda A.\lambda B.
        \Sigma\prn{\Pi\prn{A,\lambda\_.B},\lambda f.\Pi\prn{B,\lambda b.\Con{isContr}
          \prn{\Sigma\prn{A,\lambda
          a.\Con{path}\prn{\lambda\_.B,f\prn{a},b}}}}}
    \\[6pt]
    \Con{unglue} &: \Params{\phi,B,A,f} \IHom{\Tm\prn{\Con{glue}\prn{\phi,B,A,f}}}{\Tm\prn{B}}\\
    \Con{unglue} &:= \lambda g.\pi_2\prn{{\Con{glue}/\Con{tm}}\Sup{-1}\prn{g}}
  \end{align*}

  \caption{Generating clauses in the signature for Cartesian cubical type theory pertaining to connectives. For space reasons, we omit the computation rules of the Kan operations for each connective, which can be found in \cite{abcfhl:2019}.}
  \label{fig:cubical-tt:connectives}
\end{figure*}

\subsection{LCCCs as a logical framework}

The primary aspects of a type theory are its judgments $A\ \mathsf{type}$ and $a
: A$ and their derivations, many of which require hypothetical judgments (e.g.,
$\lambda x.b:A\to B$ when $x:A\vdash b:B$). One typically restricts which
judgments may be hypothesized, allowing $\prn{x:A}$ but not $\prn{X\
\mathsf{type}}$, judgmental equalities $\prn{a=b:A}$, or hypothetical judgments.
These restrictions, realized by the notion of \emph{context}, are crucial to the
syntactic metatheorems on which implementations rely; for example, decidability
of equality requires that intensional type theory lacks a context former
(isomorphic to) $\gls{\Gamma, a = b : A}$.

Both the rule-based presentations and the common categorical semantics of type
theory---including categories with families~\citep{dybjer:1996}, natural
models~\citep{awodey:2018:natural-models}, and Uemura's recent generalization
thereof known as representable map categories~\citep{uemura:2019}---include a
notion of context as part of the definition of a type theory, and require these
contexts to be preserved by models and their homomorphisms.

In contrast, higher-order logical frameworks for defining type theories, such as
Martin-L\"of's logical framework~\citep{nordstrom-peterson-smith:1990} and the
Edinburgh Logical Framework~\citep{harper-honsell-plotkin:1993}, elevate the
\emph{judgmental} structure of a type theory; then, as explicated by
\citet{harper-honsell-plotkin:1993}, one may impose after the fact a collection
of LF contexts (or \emph{worlds}) relative to which adequacy and other
metatheorems hold \citep{harper-licata:2007}. These worlds, which can be seen to
correspond roughly to the \emph{arities} of \citet{jung-tiuryn:1993}, were
subsequently implemented in the Twelf proof
assistant~\citep{pfenning-schuermann:1999} as ``\texttt{\%worlds} declarations.''

In light of that perspective, we regard a notion of context as a structure
placed on a locally Cartesian closed \emph{category of judgments} $\CatIdent{T}$
of a type theory, whose objects and morphisms are (equivalence classes of)
judgments and deductions. The dependent products of $\CatIdent{T}$ encode
hypothetical judgments, and the finite limits both encode substitution and
judgmental equality; a notion of context is often a full subcategory
$\CatIdent{C}\subseteq\CatIdent{T}$ spanned by objects distinguished as contexts.

\begin{example}\label{eg:bare-mltt}
  The category of judgments $\CatIdent{T}$ of Martin-L\"of type theory without
  any types is the free LCCC generated by a single map $\Mor{\Tm}{\Tp}$. The
  category of contexts $\CatIdent{C}\subseteq\CatIdent{T}$ is inductively
  defined as the full subcategory spanned by the terminal object and any fiber
  of $\Mor{\Tm}{\Tp}$ over a context.

  Equality is undecidable in $\CatIdent{T}$ (as it has all finite limits), but
  \emph{is} decidable for \emph{terms and types in context}, objectified by the
  restricted Yoneda embedding $\Mor{\CatIdent{T}}{\Psh{\CatIdent{C}}}$ taking the
  judgment $\Tp : \CatIdent{T}$ to its ``functor of context-valued points''
  $\Hom{\CatIdent{C}}{\Tp} : \Mor{\OpCat{\CatIdent{C}}}{\SET}$. Proofs of
  decidability proceed by further restricting to a category of (contexts and)
  renamings $\CatIdent{R}$ along the forgetful map
  $\Mor{\CatIdent{R}}{\CatIdent{C}}$
  \citep{altenkirch-hofmann-streicher:1995,fiore:2002,coquand:2019}.
  \qed
\end{example}

A second distinction between LCCCs and (for instance)
Uemura's framework is that the latter is stratified to ensure that the use of
hypothetical judgment in a theory is strictly positive, whereas both LCCCs and
syntactic logical frameworks place no restriction on hypothetical judgments and
even allow binders of higher-level (e.g., in Martin-L\"of's $\Con{funsplit}$
operator~\citep{nordstrom-peterson-smith:1990}). However, using an Artin
gluing argument due to Paul Taylor~\citep{taylor:1999}, Gratzer and
Sterling~\citep{gratzer-sterling:2020} have observed that the extension of a
representable map category \`a la Uemura to a LCCC is conservative (in fact,
fully faithful), ensuring the adequacy of LCCC encodings of type theories.

We therefore define the category of judgments of Cartesian cubical type theory
as the LCCC $\ThCat$ generated by the signature in
\cref{fig:cubical-tt:judgments,fig:cubical-tt:connectives}; then, locally
Cartesian closed functors $\SynAlg : \Mor{\ThCat}{\CatIdent{E}}$ determine
algebras for that signature valued in $\CatIdent{E}$. Unlike homomorphisms of
models of type theory, such functors preserve higher-order judgments; note
however that we are proving a single theorem about the syntactic category
$\ThCat$, not studying the model theory of cubical type theory.

\subsection{The signature of Cartesian cubical type theory}

In \cref{fig:cubical-tt:judgments} we present the judgmental structure of
cubical type theory; we inherit from Martin-L\"of type theory the basic forms of judgment
$\Tp:\ThCat$ and $\Tm:\Sl*{\ThCat}{\Tp}$ classifying types and terms respectively, and add
three additional forms of judgment for cubical phenomena: $\II:\ThCat$ for elements of the interval,
$\FF:\ThCat$ for \emph{cofibrations} (or \emph{cofibrant propositions}), and
$\brk{-}:\Sl*{\ThCat}{\FF}$ for proofs of cofibrations. The standard notion of
context is generated by $\ObjTerm{}$ and context extension by $\prn{x:A}$,
$\prn{i:\II}$, and $\prn{\_:\brk{\phi}}$.%
  \footnote{Syntactic presentations typically write $\gls{\Gamma,\phi}$ for
  $\gls{\Gamma,\_:\brk{\phi}}$.}
In this paper, we will only consider a more restricted notion of \emph{atomic
context} (\cref{sec:atomic}) that plays a role analogous to the renamings of
\cref{eg:bare-mltt}.

Cofibrations are strict propositions; the cofibration classifier $\FF$ is
strictly univalent and closed under $\land_\FF$, $\lor_\FF$, $\forall_\II$, and
$=_\II$. We define $\land_\FF$, $\forall_\II$, and $=_\II$ in terms of the same
notions in $\ThCat$, but $\ThCat$ has no disjunction or empty type by which to
define $\lor_\FF$ or stipulate $\IHom{\prn{0=1}}{\bot}$. Instead, we axiomatize
these by eliminators $\Con{abort}\Sub{\Kwd{J}},\Con{split}\Sub{\Kwd{J}}$ for
each generating judgment $\Kwd{J}$.

The following notations in \cref{fig:cubical-tt:judgments} are used pervasively
throughout this paper. (At present, ``propositions'' are elements of $\FF$;
later they will be elements of a subobject classifier $\Omega$.)

\begin{notation}[Extent types \citep{riehl-shulman:2017}]
  Given a proposition $\phi$ and a \emph{partial element} $a_\phi :
  \IHom{\brk{\phi}}{A}$, we write $\Ext{A}{\phi}{a_\phi}$ for the collection of
  elements $a : A$ that restrict to $a_\phi$ under the assumption of
  $z:{\brk{\phi}}$. In other words:
  \[
    \Ext{A}{\phi}{a_\phi} := \Compr{a:A}{\forall z : \brk{\phi}. a = a_\phi\prn{z}}
  \]
  We will implicitly coerce elements of $\Ext{A}{\phi}{a_\phi}$ to $A$.
\end{notation}

\begin{notation}[Systems \citep{cchm:2017}]
  Let $\phi,\psi$ be propositions. Under $\_:\brk{\phi\lor\psi}$, given a pair
  of partial elements $a_\phi : \IHom{\brk{\phi}}{A}$ and $a_\psi :
  \IHom{\brk{\psi}}{A}$ that agree when $\_:\brk{\phi\land\psi}$, we write
  \[
    \brk{\phi\hookrightarrow a_\phi, \psi\hookrightarrow a_\psi} : A
  \]
  for the ``case split'' that extends $a_\phi,a_\psi$. Likewise, under the
  assumption $\_:\brk{\bot}$, we write $\brk{}:A$ for the unique element of $A$.
  We will leave abstraction and application over $z:\phi$ implicit; where it
  improves clarity, we may write the unary system $\brk{\phi \hookrightarrow a}$
  for $\Mute{\lambda z:\phi.}a$.
\end{notation}

In \cref{fig:cubical-tt:connectives} we define the connectives of cubical type
theory. We specify the elements of $\Pi$, $\Sigma$, $\Con{path}$, and
$\Con{glue}$ by isomorphisms whose underlying functions encode introduction and
elimination rules, and whose equations encode $\beta$ and $\eta$ rules; we will
leave the first three of these isomorphisms (and the projection from
equivalences to functions) implicit. The higher inductive circle $\Con{S1}$ has
constructors $\Con{base}$ and $\Con{loop}$ and an eliminator
$\Con{ind}\Sub{\Con{S1}}$ with computation rules. Finally, we specify the Kan
operations $\Con{hcom}$ and $\Con{coe}$; for space reasons we do not reproduce
the computation rules for $\Con{hcom}$ and $\Con{coe}$ in each type, which can
be found in \citet{abcfhl:2019}.

The signature for De Morgan cubical type theory
\citep{cchm:2017,coquand-huber-mortberg:2018} differs only in the structure
imposed on $\II$ and the types and computation rules of $\Con{hcom}$ and
$\Con{coe}$.

\section{Synthetic Tait computability}\label{sec:stc}

\begin{figure}
  \begin{center}
    \begin{tabular}{ll}
      \toprule
      \textbf{Axiom} & \textbf{Substantiation}\\
      \midrule
      \cref{ax:universes} & \cref{lem:gluing-psh}\\
      \cref{ax:interval} & \cref{con:glued-interval,cor:interval-tiny}\\
      \cref{ax:syn} & \cref{cmp:syntactic-open}\\
      \cref{ax:syntactic-algebra} & \cref{rem:syntactic-algebra,con:glued-interval}\\
      \cref{ax:locality} & \cref{lem:0=1-implies-syn}\\
      \cref{ax:var}&\cref{con:var,rem:glued-variables}\\
      \bottomrule
    \end{tabular}
  \end{center}

  \caption{A dictionary between the axioms of \cref{sec:stc} and their substantiations in the category of computability structures.}
  \label{fig:stc:dictionary}
\end{figure}

In this section, we axiomatize a category $\CatSTC$ whose internal language
provides a ``type theory of proof-relevant logical relations'' \`{a} la Sterling
and Harper's synthetic Tait computability \citep{sterling-harper:2020}. Inside
that type theory, we then define a reflection--reification computability model
of cubical type theory from which we will derive normalization. We defer to
\cref{sec:computability-topos} an explicit construction of $\CatSTC$ as a
category of computability structures; \cref{fig:stc:dictionary} provides forward
references to our justifications of the axioms.

We begin by assuming that $\CatSTC$ satisfies Giraud's axioms \citep{sga:4}; all
such categories interpret extensional Martin-L\"of type theory extended by a
strictly univalent universe $\Omega$ of all proof-irrelevant propositions.
Next, we assume that $\CatSTC$ contains a cumulative hierarchy of universes
whose elements satisfy a strictification axiom introduced by \citet{bbcgsv:2016,orton-pitts:2016}.

\begin{definition}\label{def:op-universe}
  An \emph{strong universe} is a type theoretic universe $\UU$ strictly closed
  under dependent products, dependent sums, inductive types, quotients, and the subobject
  classifier, such that the following additional \emph{strict gluing} axiom
  holds:
  \begin{quote}
    Given $A:\UU$, write $\Con{Iso}\prn{A} := \Sum{B:\UU}\prn{A\cong B}$
    for the type of $\UU$-isomorphs of $A$. For any proposition $\phi:\Omega$,
    there is a section to the weakening map
    $\Mor{\Con{Iso}\prn{A}}{\IHom*{{\phi}}{\Con{Iso}\prn{A}}}$.
  \end{quote}
\end{definition}

\begin{axiom}\label{ax:universes}
  There exists a cumulative hierarchy of strong universes
  $\UU_0\subseteq\UU_1\ldots$ in $\CatSTC$ such that every map in $\CatSTC$ is
  classified by some $\UU_i$.
\end{axiom}

We schematically write $\UU,\VV$ for arbitrary universes in the hierarchy
specified by \cref{ax:universes}.

\begin{axiom}\label{ax:interval}
  There exists a tiny~\citep{lawvere:1980,yetter:1987} interval object $\II:\UU$
  with two endpoints $0,1:\II$.
\end{axiom}

What it means for the interval to be \emph{tiny} is that the exponential functor
$\Mor[\IHom*{\II}{-}]{\CatSTC}{\CatSTC}$ has a right adjoint $\Root*{-}$.
Equivalently, the exponential functor preserves colimits.

\subsection{Modalities for syntax and semantics}\label{sec:stc:modalities}

The central assumption of synthetic Tait computability is the existence of an
uninterpreted proposition $\Syn$ from which we will generate the modal
syntax--semantics duality.

\begin{axiom}\label{ax:syn}
  There exists a proposition $\Syn : \Omega$.
\end{axiom}

The proposition $\Syn$ generates complementary \emph{open} and \emph{closed} lex
idempotent modalities $\Op,\Cl$ that we interpret as respectively projecting the
syntactic and semantic aspects of a given computability structure. Because lex
modalities descend to the slices in a fibered way, we can use $\Op,\Cl$
na\"ively in the internal language of $\CatSTC$
\citep{bmss:2011,rijke-shulman-spitters:2017}; however, we find it most
convenient to begin by considering the universes $\UUSyn,\UUSem$ of
``syntactic'' and ``semantic'' types.

\subsubsection{Universe of syntactic types}

Given a universe $\UU:\VV$, we define the universe $\UUSyn:\VV$ of syntactic
types together with its (dependent) modality; the following definitions are
justified by strict gluing (\cref{ax:universes}), setting
$\phi:=\Syn$.
\begin{gather*}
  \begin{aligned}
    \UUSyn &: \Ext{\VV}{\Syn}{\UU}\\
    \UUSyn &\cong \IHom{\Syn}{\UU}
  \end{aligned}
  \quad
  \begin{aligned}
    \Op &: \Ext{\IHom{\UU}{\UUSyn}}{\Syn}{\lambda A.A}\quad\\
    \Op &\cong \lambda A.\lambda\_:\Syn.A
  \end{aligned}
  \\
  \begin{aligned}
    \Con{el}\Sub{\Op} &: \Ext{\IHom{\UUSyn}{\UU}}{\Syn}{\lambda A.A}\\
    \Con{el}\Sub{\Op} &\cong \lambda A. \Prod{z:\Syn}{A\prn{z}}
  \end{aligned}
\end{gather*}

To see how strictification applies, observe that $\IHom{\Syn}{\UU}$ is an
isomorph of $\UU$ under the assumption $\_:\Syn$; we may therefore choose
$\UUSyn$ to be (totally) isomorphic to $\IHom*{\Syn}{\UU}$ and under $\_:\Syn$
strictly equal to $\UU$. The remaining definitions go through directly given
that $\UUSyn$.

Because it causes no ambiguity, we will leave the decoding $\Con{el}\Sub{\Op}$
implicit in our notations; furthermore, we will leave both abstraction and
application over $\Syn$ implicit.

Our use of strict gluing above can be summed up in the following
syntactic realignment lemma, the workhorse of synthetic Tait computability.

\begin{restatable}[Syntactic realignment~\citep{sterling-harper:2020,sterling-gratzer:2020:stc,sterling-angiuli-gratzer:2020}]{corollary}{CorRealignment}\label{cor:realignment}
  Given $A:\UU$, $A_\circ : \UUSyn$, and an isomorphism $f :
  \Op\prn{A_\circ\cong A}$, we may define a strictly aligned type $f^*A :
  \Ext{\UU}{\Syn}{A_\circ}$ and a strictly aligned isomorphism $f^\dagger :
  \Ext{f^*A\cong A}{\Syn}{f}$.
\end{restatable}

\begin{remark}
  The syntactic modality commutes with
  dependent products, dependent sums, equality, etc.
\end{remark}

Then we axiomatize the existence of an algebra for the signature of Cartesian
cubical type theory in $\CatSTC$. One can internalize as a dependent record the
collection $\ThMod{\VV}$ of $\ThCat$-algebras/models valued in types classified
by any universe $\VV$, writing $\SynAlg.\Tp$, $\SynAlg.\Tm$, etc.\ for each
component. In \cref{ax:syntactic-algebra} below, we require a $\ThCat$-model $\SynAlg$ valued in $\UUSyn$, such that
$\SynAlg.\II$ is the syntactic part of the interval of \cref{ax:interval}.

\begin{axiom}\label{ax:syntactic-algebra}
  There exists a $\ThCat$-model $\SynAlg : \ThMod{\UUSyn}$
  such that $\Op\prn{\SynAlg.\II = \II}$.
\end{axiom}

\subsubsection{Universe of semantic types}

A type $A:\UU$ is called \emph{semantic} (or \emph{$\Op$-connected})
when it has no syntactic component, i.e., we have an isomorphism
$\ObjTerm{}\cong \Op{A}$. Using this idea as a prototype, we define a dual
universe of semantic types:
\begin{gather*}\small
  \begin{aligned}
    \UUSem &: \Ext{\VV}{\Syn}{\ObjTerm{}}\\
    \UUSem &\cong \Ext{\UU}{\Syn}{\ObjTerm{}}
  \end{aligned}
  \quad
  \begin{aligned}
    \Cl &: \IHom{\UU}{\UUSem}\\
    \Cl &\cong \lambda A. A\sqcup_{A\times\Syn}\Syn
  \end{aligned}
  \quad
  \begin{aligned}
    \Con{el}\Sub{\Cl} &: \IHom{\UUSem}{\UU}\\
    \Con{el}\Sub{\Cl} &\cong \lambda A. A
  \end{aligned}
\end{gather*}

Above we are writing $A\sqcup_{A\times\Syn}\Syn$ for the pushout of the two product projections from $A\times\Syn$.

The definitions of $\UUSem, \Cl, \Con{el}\Sub{\Cl}$ are likewise justified by
syntactic realignment (\cref{cor:realignment}): fixing $\_:\Syn$ we note that
each of the types above becomes a singleton, so it can be aligned to restrict
to $\ObjTerm{}$ strictly. As with the syntactic modality, we leave the decoding
$\Con{el}\Sub{\Cl}$ implicit.

\begin{warn}
  The semantic modality commutes with dependent sums and equality, but not much else.
\end{warn}

\subsection{Cofibrations and locality}\label{sec:cof}

We construct the universe of cofibrations in two steps: first we define a
universe of propositions $\ExtFF : \Ext{\UU}{\Syn}{\SynAlg.\FF}$, and then we
constrain it to a subclass $\FF\subseteq\ExtFF$ generated by equality of the
interval, disjunction, conjunction, and universal quantification over the
interval. Constraining $\FF$ in this way will allow us to define an external
algorithm to decide equality underneath a cofibration (\cref{thm:dec-eq}).
\begin{align*}
  \ExtFF &: \Ext{\UU}{\Syn}{\SynAlg.\FF}\\
  \ExtFF
    &\cong \Sum{
      \phi:\SynAlg.\FF
    }{
      \Ext{\Omega}{\Syn}{\SynAlg.\brk{\phi}}
    }
\end{align*}

It is trivial to close $\ExtFF$ under conjunction, equality of the interval, and
universal quantification over the interval.

\begin{notation}
  The canonical map $\prn{- = \top_\ExtFF} : \IHom{\ExtFF}{\Omega}$ is a
  suitable decoding function; we treat it as an implicit coercion.
\end{notation}

As in \cref{fig:cubical-tt:judgments}, the difficult part is to close $\ExtFF$
under disjunction and to enforce $0\neq 1$; because
$\SynAlg.\Con{split}\Sub{\Kwd{J}}$ only eliminates into components of $\SynAlg$,
the ``disjunction'' $\SynAlg.\prn{\lor_\ExtFF}$ is not even a disjunction relative
to types in $\UUSyn$, much less in all of $\UU$ (and similarly for
$\SynAlg.\Con{abort}\Sub{\Kwd{J}}$ and $\SynAlg.\FakeFalse$).

\begin{construction}[Disjunction]\label{con:disj}
  We explicitly glue together the syntactic disjunction with the semantic
  disjunction; to ensure that the resulting proposition is aligned over the
  (weaker) syntactic disjunction, we place the semantic disjunction underneath
  the modality $\Cl$ to force it to become $\Op$-connected.
  \begin{align*}
    \prn{\lor_\ExtFF} &: \Ext{\IHom{\ExtFF\times\ExtFF}{\ExtFF}}{\Syn}{\SynAlg.\prn{\lor_\FF}}\\
    \phi\lor_\ExtFF\psi &=
    \prn{
      \phi\mathbin{\SynAlg.\lor_\FF}\psi,
      \SynAlg.\brk{\phi\mathbin{\SynAlg.\lor_\FF}\psi}
      \land
      \Cl\prn{
        \brk{\phi}\lor\brk{\psi}
      }
    }
  \end{align*}
  No realignment is required, because $\Omega$ is strictly univalent.
\end{construction}

We may then define the universe of cofibrations $\FF\subseteq\ExtFF$ to be the
smallest subobject of $\ExtFF$ closed under the following rules:
\begin{mathpar}
  \inferrule{
    z:\Syn\\
    \phi : \ExtFF
  }{
    \phi\in\FF
  }
  \and
  \inferrule{
    \phi\in\FF
    \\
    \psi\in\FF
  }{
    \phi\land_\ExtFF \psi\in\FF
    \qquad
    \phi\lor_\ExtFF \psi\in\FF
  }
  \and
  \inferrule{
    \forall i:\II. \prn{\phi\prn{i}\in\FF}
  }{
    \prn{\forall i.\phi\prn{i}}\in \FF
  }
  \and
  \inferrule{
    r,s:\II
  }{
    \prn{r = s}\in \FF
  }
\end{mathpar}

To avoid confusion, we will write $\land_\FF,\lor_\FF : \FF\times\FF\to\FF$ for
the maps induced by the closure of $\FF$ under $\land_\ExtFF,\lor_\ExtFF$. Likewise we
define $\top_\FF = (1=1)$ and $\FakeFalse = (0=1)$ in $\FF$. We observe that
the universe of cofibrations can be aligned over $\SynAlg.\FF$, i.e., we have
$\FF : \Ext{\UU}{\Syn}{\SynAlg.\FF}$; this follows from the fact that
$\prn{\phi\in\FF} =_{\Omega}\top$ for any $\phi:\ExtFF$ assuming $z:\Syn$.

We must define semantic conditions for types that are \emph{local} with respect to
$\lor_\FF$ and $\FakeFalse$, in the sense that they behave as though the
positive cofibrations satisfy universal properties.

\begin{definition}\label{def:01-connected}
  A type $A:\UU$ is called \emph{$\FakeFalse$-connected} when it behaves as if
  $0\not=1$, i.e., we have $\IHom*{\brk{\FakeFalse}}{A} \cong \ObjTerm$.
\end{definition}

\begin{definition}\label{def:local}
  A type $A:\UU$ is called \emph{$\FF$-local} when it is $\FakeFalse$-connected and,
  for any two cofibrations $\phi,\psi : \FF$, the type $A$ is right-orthogonal
  to the canonical map $\Mor{\brk{\phi}\lor\brk{\psi}}{\brk{\phi\lor_\FF\psi}}$
  in the sense that the dotted map below always exists uniquely:
  \[
    \begin{tikzpicture}[diagram]
      \SpliceDiagramSquare{
        nw = \brk{\phi}\lor\brk{\psi},
        sw = \brk{\phi\lor_\FF\psi},
        se = \ObjTerm,
        ne = A,
        north = \text{\scriptsize{$\brk{\phi\hookrightarrow a_\phi, \psi\hookrightarrow a_\psi}$}},
        width = 4cm,
      }
      \path[->,exists] (sw) edge node [desc] {\scriptsize{$\brk{\phi\hookrightarrow a_\phi, \psi\hookrightarrow a_\psi}_A$}} (ne);
    \end{tikzpicture}
  \]
\end{definition}

A necessary condition for a type $A:\UU$ being $\FF$-local is that its syntactic
part $\Op{A}$ is $\FF$-local. \cref{ax:locality} below implies that this
condition is sufficient, unfolding \cref{con:disj}.

\begin{axiom}\label{ax:locality}
  We have $\bot_\FF\leq\Syn$ in the internal logic.
\end{axiom}

\vspace{-0.6em} %
\begin{restatable}{corollary}{LemLocality}\label{lem:locality}
  A type $A:\UU$ in $\CatSTC$ is $\FF$-local if and only if $\Op{A}$ is $\FF$-local.
\end{restatable}

The interpretation of every syntactic sort of Cartesian cubical type theory in
$\SynAlg$ can be seen to be $\FF$-local. Therefore, by \cref{ax:locality}, any
type $A$ whose syntactic part $\Op{A}$ is isomorphic to one of those sorts is
automatically $\FF$-local.

\subsection{Kan computability structures}

\begin{definition}
  We define $\UUTp$ to be the object of \emph{computability structures}, which
  pair a syntactic type $A:\SynAlg.\Tp$ with a total type
  aligned over its elements:
  \begin{align*}
    \UUTp &: \Ext{\VV}{\Syn}{\SynAlg.\Tp}\\
    \UUTp &\cong \Sum{
      A : \SynAlg.\Tp
    }{
      \Ext{\UU}{\Syn}{\SynAlg.\Tm\prn{A}}
    }
  \end{align*}
  We leave the projection $\Mor{\UUTp}{\UU}$ implicit in our notation.
\end{definition}

\begin{figure*}
  \begin{align*}
    \Con{HCom}&: \Mor{\UUTp}{\UUSem}\\
    \Con{Coe}&: \Mor{\IHom*{\II}{\UUTp}}{\UUSem}\\
    \Con{HCom}\prn{A} &\cong
    \Ext{
      \Prod{r,s:\II}{
        \Prod{\phi:\mathbb{F}}{
          \Prod{
            a:
            \Prod{i:\II}{
              \Prod{\_:\brk{i=r\lor_\FF\phi}}{
                A
              }
            }
          }{
            \Ext{A}{r=s\lor_\FF\phi}{
              a\prn{s}
            }
          }
        }
      }
    }{\Syn}{
      \SynAlg.\Con{hcom}_{A}
    }
    \\
    \Con{Coe}\prn{A} &\cong
    \Ext{
      \Prod{r,s:\II}{
        \Prod{a:A\prn{r}}{
          \Ext{A\prn{s}}{r=s}{a}
        }
      }
    }{\Syn}{\SynAlg.\Con{coe}\Sub{A}}
  \end{align*}
  \caption{Definitions of homogeneous composition structures for a
  computability structure and coercion structures for a line of computability
  structures; the use of $\UUSem$ indicates that these structures are
  \emph{$\Op$-connected}, i.e., equivalently classified by the
  subuniverse $\Ext{\UU}{\Syn}{\ObjTerm}$.}
  \label{fig:kan-structures}
\end{figure*}

\begin{definition}
  A \emph{homogeneous composition structure} on $A : \UUTp$ is an element of the
  type $\Con{HCom}\prn{A}$ defined in \cref{fig:kan-structures}; such a
  structure asserts the existence of an operation $\Con{hcom}_A$ that is aligned
  over the existing syntactic homogeneous composition operation. We define by
  realignment a weak classifying object $\UUTpHCom$ for computability structures
  equipped with a homogeneous composition structure:
  \begin{align*}
    \UUTpHCom &: \Ext{\VV}{\Syn}{\SynAlg.\Tp}\\
    \UUTpHCom &\cong \Sum{A : \UUTp}{\Con{HCom}\prn{A}}
  \end{align*}
  We leave the projection $\Mor{\UUTpHCom}{\UUTp}$ implicit.

  Likewise, a \emph{coercion structure} on a line of computability structures $A
  : \IHom{\II}{\UUTp}$ is an element of the type $\Con{Coe}\prn{A}$ also defined
  in \cref{fig:kan-structures}; such a structure provides a $\Con{coe}\Sub{A}$
  operation aligned over the existing syntactic coercion operation.
\end{definition}

Constructing a weak classifying object for coercion structures is more
challenging; we use the method of \citet{lops:2018}, which relies crucially on
the tininess of the interval.

\begin{construction}
  Using the right adjoint $\Root*{-}$ to $\IHom*{\II}{-}$ given by
  \cref{ax:interval}, we transpose the map
  $\Mor[\Con{Coe}]{\IHom*{\II}{\UUTpHCom}}{\UUSem}$ to obtain
  $\Mor[\Con{Coe}^\sharp]{\UUTpHCom}{\Root*{\UUSem}}$. Pulling back the ``root''
  of the generic family $\Mor{{\dot\UU}_{\Cl}}{\UUSem}$ along this map, we
  obtain a weak classifying object $\UUTpKan$ for computability structures with
  both homogeneous composition and coercion, which we call \emph{Kan}.
  \[
    \DiagramSquare{
      nw/style = pullback,
      nw = \UUTpKan,
      ne = \Root*{{\dot\UU}_{\Cl}},
      se = \Root*{\UUSem},
      sw = \UUTpHCom,
      south = \Con{Coe}^\sharp,
      width = 2.5cm,
      height = 1.5cm,
    }
  \]

  Because all the structures we are adding to $\UUTp$ remain $\Op$-connected, we
  may align this pullback as a large type $\UUTpKan :
  \Ext{\VV}{\Syn}{\SynAlg.\Tp}$. The left map $\Mor{\UUTpKan}{\UUTpHCom}$
  projects a homogeneous composition operation $\Con{hcom}_A$ for every Kan
  computability structure $A:\UUTpKan$; transposing the upstairs map, we see
  that each line of Kan computability structures $A:\IHom{\II}{\UUTpKan}$ also
  carries a coercion structure $\Con{coe}\Sub{A}$. We leave the composite
  projection $\Mor{\UUTpKan}{\UUTp}$ implicit.
\end{construction}

\subsection{Neutral and normal forms}

In order to axiomatize the neutral and normal forms, we will need a
computability structure of \emph{term variables}.

\begin{axiom}\label{ax:var}
  We assume a family of types $\Con{var} : \Prod{A:\SynAlg.\Tp}{\Ext{\UU}{\Syn}{\SynAlg.\Tm\prn{A}}}$.
\end{axiom}

\begin{figure}
  \begin{align*}
    \NfTp &: \Ext{\UU}{\Syn}{\SynAlg.\Tp}\\
    \Ne{\prn{-}} &: \Prod{\phi:\FF}{\Prod{A:\SynAlg.\Tp}{\Ext{\UU}{\phi\lor\Syn}{\SynAlg.\Tm\prn{A}}}}\\
    \Nf &: \Prod{A:\SynAlg.\Tp}{\Ext{\UU}{\Syn}{\SynAlg.\Tm\prn{A}}}\\
    \frk{var} &: \Params{A} \Ext{\IHom{\Con{var}\prn{A}}{\Ne{\bot_\FF}{A}}}{\Syn}{\lambda x.x}
  \end{align*}
  \caption{Axiomatization of the structure of normal and neutral forms.}
  \label{fig:ne-nf}
\end{figure}

\begin{figure*}[ht]
  \small
  \begin{nf-box}
    \begin{align*}
      \frk{path} &: \Ext{
        \IHom{
          \prn{
            \Sum{A : \IHom{\II}{\NfTp}}{
              \Nf{A\prn{0}}\times\Nf{A\prn{1}}
            }
          }
        }{\NfTp}
      }{\Syn}{
        \SynAlg.\Con{path}
      }
      \\
      \frk{plam} &:
      \Params{A}
      \Ext{
        \Prod{
          p :
          \Prod{i:\II}{
            \Nf{A\prn{i}}
          }
        }{
          \Nf{\SynAlg.\Con{path}\prn{A,p(0),p(1)}}
        }
      }{\Syn}{\lambda p.\lambda i.p\prn{i}}
      \\
      \frk{papp} &:
      \Params{\phi,A,a_0,a_1}
      \Ext{
        \IHom{
          \Ne{\phi}{\SynAlg.\Con{path}\prn{A,a_0,a_1}}
        }{
          \Prod{i:\II}{
            \Ne{\phi\lor_\FF\partial{i}}{A\prn{i}}
          }
        }
      }{\Syn}{\lambda p.\lambda i.p\prn{i}}
    \end{align*}
  \end{nf-box}

  \begin{align*}
    \Con{path} &: \Ext{
      \IHom{\prn{\Sum{A : \IHom{\II}{\Tp}} A\prn{0}\times A\prn{1}}}{\Tp}
    }{\Syn}{\SynAlg.\Con{path}}
    \\
    \brk{\Con{path}\prn{A,a_0,a_1}} &\cong
    \Prod{i:\II}{
      \Ext{A\prn{i}}{\partial{i}}{
        \brk{\overline{i=\epsilon\hookrightarrow a_\epsilon}}\Sub{\SynAlg.\Tm\prn{A\prn{i}}}
      }
    }
    \\
    \Con{hcom}\Sub{\Con{path}\prn{A,a_0,a_1}}\Sup{r\to s;\phi}p &=
    \lambda i.
    \Con{hcom}\Sub{A\prn{i}}\Sup{r\to s; \phi\lor_\FF\partial{i}}
    \lambda k.
    \brk{
      k=r\lor_\FF\phi\hookrightarrow p\prn{k},
      \overline{i=\epsilon\hookrightarrow a_\epsilon}
    }\Sub{\SynAlg.\Tm\prn{A\prn{i}}}
    \\
    \Con{coe}\Sub{\lambda j.\Con{path}\prn{\lambda i.A\prn{i,j},a_0\prn{j},a_1\prn{j}}}\Sup{r\to s}p &=
    \lambda i.
    \Con{com}\Sub{\lambda j.A\prn{i,j}}\Sup{r\to s; \partial{i}}
    \lambda j.\brk{
      j=r\to p\prn{i},
      \overline{i=\epsilon\hookrightarrow a_\epsilon\prn{j}}
    }\Sub{\SynAlg.\Tm\prn{A\prn{i,j}}}
    \\
    \NormTp{\Con{path}\prn{A,a_0,a_1}} &=
    \frk{path}\prn{\lambda i. \NormTp{A\prn{i}}, \Reify{A\prn{0}}a_0, \Reify{A\prn{1}}a_1}
    \\
    \Reflect{\Con{path}\prn{A,a_0,a_1}}{\phi}\NeGlue{p_0}{\phi}{p}
    &=
    \lambda i. \Reflect{A\prn{i}}{\phi\lor_\FF\partial{i}}{
      \NeGlue{\frk{papp}\prn{p_0,i}}{\phi\lor_\FF\partial{i}}{
        \brk{
          \phi\hookrightarrow p\prn{i},
          \overline{i=\epsilon\hookrightarrow a_\epsilon}
        }\Sub{\SynAlg.\Tm\prn{A\prn{i}}}
      }
    }
    \\
    \Reify{\Con{path}\prn{A,a_0,a_1}}p &=
    \frk{plam}\prn{
      \lambda i.
      \Reify{A\prn{i}}p\prn{i}
    }
  \end{align*}

  \caption{The cubical normalization structure for dependent path types.}
  \label{fig:norm-path}
\end{figure*}

\begin{figure*}
  \small
  \begin{nf-box}
    \scriptsize
    \begin{align*}
      \frk{glue} &:
      \Params{\phi}
      \Ext{
        \IHom{
          \prn{
            \Sum{B:\NfTp}
            \Sum{A : \IHom{\brk{\phi}}{\NfTp}}
            \Prod{\_:\brk{\phi}}{\Nf{\Con{Equiv}\prn{A,B}}}
          }
        }{
          \Con{nftp}
        }
      }{\Syn\lor\phi}{
        \brk{
          \Syn\hookrightarrow\SynAlg.\Con{glue}\prn{\phi,-},
          \phi\hookrightarrow \lambda\prn{B,A,f}.A
        }
      }
      \\
      \frk{englue} &:
      \Params{\phi,B,A,f}
      \Ext{
        \IHom{
          \prn{
            \Sum{a : \Prod{\_:\brk{\phi}}{\Nf\prn{A}}}{
              \Ext{\Nf\prn{B}}{\Syn\land\phi}{f\prn{a}}
            }
          }
        }{
          \Nf{\SynAlg.\Con{glue}\prn{\phi,B,A,f}}
        }
      }{\Syn\lor\phi}{
        \lambda\prn{a,b}.
        \brk{
          \Syn\hookrightarrow \SynAlg.\Con{glue/tm}\prn{a,b},
          \phi\hookrightarrow a
        }
      }
      \\
      \frk{unglue} &:
      \Params{\phi,B,A,f,\psi}
      \Ext{
        \IHom{\Ne{\psi}{\SynAlg.\Con{glue}\prn{\phi,B,A,f}}}{\Ne{\psi\lor_\FF\phi}{B}}
      }{\Syn}{
        \SynAlg.\Con{unglue}\prn{\phi,-}
      }
    \end{align*}
  \end{nf-box}

  \begin{align*}
    \Con{glue} &:
    \Ext{
      \Prod{\phi:\FF}
      \Ext{
        \IHom{
          \prn{
            \Sum{B:\Tp}
            \Sum{A : \IHom{\brk{\phi}}{\Tp}}
            \Prod{\_:\brk{\phi}}{
              \Con{Equiv}\prn{A,B}
            }
          }
        }{
          \Tp
        }
      }{\phi}{
        \lambda \prn{B,A,f}. A
      }
    }{\Syn}{\SynAlg.\Con{glue}}
    \\
    \brk{\Con{glue}\prn{\phi,B,A,f}} &\cong
    \Sum{a : \Prod{\_:\brk{\phi}}{A}}{
      \Ext{B}{\phi}{f\prn{a}}
    }
    \\
    \Con{hcom}\Sub{\Con{glue}\prn{\phi,B,A,f}}\Sup{r\to s; \psi}p &=
    \Mute{
      \gls{\text{omitted for brevity, see \cite{abcfhl:2019}}}
    }
    \\
    \Con{coe}\Sub{
      \lambda i.\Con{glue}\prn{\phi\prn{i},B\prn{i},A\prn{i},f\prn{i}}
    }\Sup{r\to s}f &=
    \Mute{
      \gls{\text{omitted for brevity, see \cite{abcfhl:2019}}}
    }
    \\
    \NormTp{\Con{glue}\prn{\phi,B,A,f}} &=
    \frk{glue}\prn{\phi,\NormTp{B},\NormTp{A},\Reify{\Con{Equiv}\prn{A,B}}{f}}
    \\
    \Reflect{\Con{glue}\prn{\phi,B,A,f}}{\psi} \NeGlue{g_0}{\psi}{g} &=
    \prn{
      \brk{
        \phi\hookrightarrow
        \Reflect{A}{\psi}{
          \NeGlue{g_0}{\psi}{g}
        }
      }
      ,
      \Reflect{B}{\psi}{
        \NeGlue{\frk{unglue}\prn{g_0}}{\psi\lor_\FF\phi}{
          \brk{
            \psi\hookrightarrow \pi_2\prn{g}
            ,
            \phi \hookrightarrow
            f\prn{
              \Reflect{A}{\psi}{\NeGlue{g_0}{\psi}{g}}
            }
          }\Sub{B}
        }
      }
    }
    \\
    \Reify{\Con{glue}\prn{\phi,B,A,f}} \prn{a,b} &=
    \frk{englue}\prn{
      \brk{\phi\hookrightarrow \Reify{A}{a}},
      \Reify{B}{b}
    }
  \end{align*}

  \caption{The cubical normalization structure for glue types.}
  \label{fig:norm-glue}
\end{figure*}
 
In \cref{fig:ne-nf}, we axiomatize the judgmental structure of the normal and
neutral forms of cubical type theory in the language of synthetic Tait
computability; \cref{fig:norm-path,fig:norm-glue} contain the normal and neutral
forms of path and glue types (within dashed boxes); $\Sigma$-types, $\Pi$-types,
and the circle are located in
\ifarXiv
  the appendix.
\else
  the appendix~\citep{sterling-angiuli:2021:extended}.
\fi
The main part of the normalization argument only needs constants of the kind
listed to exist, but we substantiate these constants with an external
inductive definition in \ifarXiv\cref{sec:explicit-nf}\else{the appendix}~\citep{sterling-angiuli:2021:extended}\fi.

As discussed in \cref{sec:what-are-the-neutrals}, the main difficulty in
adapting Coquand's semantic normalization argument~\citep{coquand:2019} to
cubical type theory is that neutral terms do \emph{not} evince a
cubically-stable aspect of the syntax of cubical type theory. The simplest
example of this behavior is the neutral form of a path application
$\frk{papp}\prn{p,i}$, which is ``neutral'' in the traditional sense only so
long as $i$ is neither $0$ nor $1$.

\begin{warn}
  It is reasonable but ultimately futile to try and restrict the second
  argument of $\frk{papp}$ to be a ``variable'' of some kind---in doing so, one
  refutes either the tininess of the interval or the existence of a Tait
  \emph{reflection} operation for paths.
\end{warn}

The failure of all previous attempts to isolate the neutral forms of cubical
type theory stems ultimately from an insistence on characterizing
\emph{positively} the conditions under which a term is neutral. We have taken
the opposite perspective, by indexing the neutrals in a ``locus of instability''
$\phi:\FF$ under which they \emph{cease} to be neutral; as soon as $\phi$
becomes true, the semantic information carried by $a :\Ne{\phi}{A}$ collapses to
a point. Our negative perspective suggests a way to ``stabilize'' a neutral form
by gluing computability data onto it along its locus of instability.

\begin{definition}[Stabilized neutrals]\label{def:sne}
  Let $A:\UUTp$ be a computability structure; a \emph{stabilized neutral} is a
  pair of a neutral $a_0:\Ne{\phi}{A}$ together with a computability datum
  $a:\Prod{\_:\brk{\phi}}{\Ext{A}{\Syn}{a_0}}$ defined on its locus of
  instability. We will write $\NeGlue{a_0}{\phi}{a}$ for such pairs, and obtain
  by realignment a type family of stabilized neutrals:
  \[\small
  \begin{aligned}
    \SNe{\phi} &: \Ext{\UUTp\to\UU}{\Syn\lor\phi}{
      \brk{\Syn\hookrightarrow\SynAlg.\Tm,\phi\hookrightarrow\lambda A.A}}\\
    \SNe{\phi}\prn{A} &\cong \Sum{a_0:\Ne{\phi}{A}}{\Prod{\_:\brk{\phi}}{\Ext{A}{\Syn}{a_0}}}
  \end{aligned}
  \]
\end{definition}

\subsection{Cubical normalization structures}

We now reach the central definition of this paper, that of a \emph{cubical
normalization structure}, a notion inspired by the Tait closure
condition~\citep{tait:1967} under which neutrals can be reflected to computable
elements and computable elements can be reified to normals, as presented for
instance by \citet{coquand:2019}. Our version of the Tait reflection operation
takes \emph{stabilized} neutrals to computable elements.

\begin{definition}
  A \emph{cubical normalization structure} $A:\Tp$ consists of the following data:
  \begin{align*}
    \brk{A} &: \UUTpKan\\
    \NormTp{A} &: \Ext{\NfTp}{\Syn}{\brk{A}}\tag{normal form}\\
    \Reflect{A} &: \Prod{\phi:\FF}
    \Ext{
      \IHom{\SNe{\phi}{\brk{A}}}{\brk{A}}
    }{\Syn\lor\phi}{\lambda a.a}
    \tag{reflect}\\
    \Reify{A} &:
    \Ext{
      \IHom{\brk{A}}{\Nf{\brk{A}}}
    }{\Syn}{\lambda a.a}
    \tag{reify}
  \end{align*}

  We define by realignment the large type $\Tp$ of cubical normalization
  structures, noting that the latter three components of a cubical normalization
  structure are $\Op$-connected:
  \[
    \Tp : \Ext{\VV}{\Syn}{\SynAlg.\Tp}
  \]
\end{definition}

\begin{remark}[Vertical maps]\label{rem:vertical}
  We refer to the reflection and reification maps as \emph{vertical}, in the
  sense that they are constrained to lie over the syntactic identity function.

  The role of verticality is to ensure that reification takes computability data
  for a given term to a normal form \emph{of the same term}, etc. Likewise, our
  presentation of the neutral and normal forms use extent types to express their
  relationship to the syntactic $\ThCat$-model without escaping the internal
  language of $\CatSTC$. In this way, extent types and vertical maps play a very
  important role in synthetic Tait computability.
\end{remark}

The main result of this section is the construction of a computability algebra
for Cartesian cubical type theory; this equips each syntactic type with a Kan
computability structure, normal form, and reflection and reification maps.

\begin{theorem}\label{thm:computability-algebra}
  We have a \emph{computability} $\ThCat$-model $\CmpAlg :
  \Ext{\ThMod{\UU}}{\Syn}{\SynAlg}$ aligned over the syntactic algebra.
\end{theorem}

\begin{proof}
  We define $\CmpAlg.\Tp = \Tp$, $\CmpAlg.\Tm\prn{A} = \brk{A}$, $\CmpAlg.\II =
  \II$, and $\CmpAlg.\FF = \FF$.
  \ifarXiv
  In Figs.~\ref{fig:norm-path} to
  \ref{fig:norm-s1:ind}, we show how to close the universe of cubical
  normalization structures $\Tp$ under the connectives $\Con{path}$, $\Pi$,
  $\Sigma$, $\Con{glue}$, and $\Con{S1}$.
  \else
  In \cref{fig:norm-path,fig:norm-glue}, we show how to close the universe of cubical normalization structures $\Tp$ under the $\Con{path}$ and $\Con{glue}$ connectives; other connectives, including $\Pi$, $\Sigma$, and $\Con{S1}$ are dealt with in the appendix~\citep{sterling-angiuli:2021:extended}.
  \fi
  The fact that the resulting model is
  aligned over $\SynAlg$ follows from each of these components being aligned
  over their syntactic counterparts; in particular, each connective is aligned
  over the syntactic connective in $\SynAlg$.
\end{proof}

\section{The computability topos}\label{sec:computability-topos}

We now define the category $\CatSTC$ in which \cref{sec:stc}
takes place, as a category of sheaves on a generalized space $\GlTop$ which
combines syntax and semantics. Mirroring the modal syntax--semantics duality
introduced in \cref{sec:stc:modalities}, sheaves on $\GlTop$ function as
computability structures because they have syntactic and semantic aspects
obtained by restriction to the corresponding regions of the space.

We construct $\GlTop$ in turn by starting with a syntactic topos $\ThTop$ that
contains the generic model of Cartesian cubical type theory, and gluing it onto
a semantic topos $\RxTop$ over which the notions of variable, neutral, and
normal form are definable.

\subsection{The language of topoi}

Topoi are sometimes thought of as generalized topological spaces, and sometimes
as special kinds of categories. These perspectives are complementary, but one
avoids many notational and conceptual quagmires by distinguishing them
formally~\citep{anel-joyal:2021,vickers:2007,bunge-funk:2006}: in the tradition
of Grothendieck, the 2-category of topoi is \emph{opposite} to the 2-category
of cocomplete and finitely complete categories satisfying Giraud's exactness
axioms~\citep{sga:4}.\footnote{The 1-cells are reversed, but the 2-cells remain
the same.} This explains why the product of two topoi is given
by a \emph{tensor product} of categories; the situation is analogous to
the other dualities between geometry and algebra in mathematics, such as
locales/frames, affine schemes/commutative rings, Stone spaces/Boolean
algebras, etc.

\begin{notation}
  Given a topos $\XTop$, we will write $\Sh{\XTop}$ for its \emph{category of
  sheaves}, which is the formal avatar of $\XTop$ in the opposite category. In
  traditional parlance, morphisms of topoi go in the ``direct image'' direction,
  and morphisms of categories of sheaves go in the ``inverse image'' direction.
\end{notation}

\begin{example}[Presheaves]\label{ex:presheaves}
  Given a small category $\CCat$, the category of presheaves $\Psh{\CCat}$ is
  the category of functors $\Mor{\OpCat{\CCat}}{\SET}$. Because $\Psh{\CCat}$
  satisfies Giraud's axioms, there is a topos $\PrTop{\CCat}$ satisfying
  $\Sh{\PrTop{\CCat}} = \Psh{\CCat}$.
\end{example}

\subsection{The syntactic topos}

Recall from \cref{sec:syntax} that a model of $\ThCat$ in $\Sh{\XTop}$ is a
locally Cartesian closed functor $\Mor{\ThCat}{\Sh{\XTop}}$; when this functor
only preserves finite limits but not dependent products, we refer to it as a
\emph{pre-model}. Pre-models play an important role in the metatheory of
higher-order logic (as in the ``general models'' of \citet{henkin:1950}, later
studied in the language of topoi by \citet{awodey-butz:2000}), as well as the
metatheory of dependent type theory (as in pseudo-morphisms of
cwfs~\citep{newstead:2018,kaposi-huber-sattler:2019}).

We define the syntactic topos $\ThTop$ as the presheaf topos $\ThTop*$. In light
of Diaconescu's theorem~\citep{diaconescu:1975}, $\ThTop$ is the classifying topos of pre-models of
$\ThCat$, in the sense that morphisms of topoi $\Mor{\XTop}{\PrTop{\ThCat}}$
correspond to pre-models of $\ThCat$ in $\Sh{\XTop}$.

\subsection{The topos of cubical atomic terms}\label{sec:atomic}

Next, we define the semantic topos $\RxTop$ and an essential morphism of topoi
$\Mor[\StrMap]{\RxTop}{\ThTop}$ along which we will glue in \cref{sec:gluing};
that $\StrMap$ is essential means there is an additional left adjoint
$\StrMap_!\dashv\StrMap^*\dashv\StrMap_*$, a technical condition that will play
an important role in \cref{sec:axioms}. Intuitively, $\RxTop$ is the topos of
\emph{cubically atomic terms}, i.e., term variables and elements of the
interval; concretely, we equip $\RxTop$ with a tiny interval object $\II \cong
\StrMap^*\II: \Sh{\RxTop}$ and a fiberwise-tiny family of term variables
$\Con{var}:\Sl{\Sh{\RxTop}}{\StrMap^*\Tm}$ indexed in syntactic types.

We define $\RxTop := \RxTop*$, where $\RxCat$ is a category of \emph{cubical
atomic contexts and substitutions} whose objects $\Gamma : \RxCat$ we define
simultaneously with their decodings $\StrMap{\Gamma} : \ThCat$:
\begin{mathpar}
  \inferrule{
  }{
    \cdot : \RxCat\\
    \StrMap{\cdot} = \ObjTerm{\ThCat}
  }
  \and
  \inferrule{
    \Gamma : \RxCat\\
    \Mor[A]{\StrMap{\Gamma}}{\Tp}
  }{
    \Gamma.A : \RxCat\\
    \StrMap{\Gamma.A} = \StrMap{\Gamma}.A
  }
  \and
  \inferrule{
    \Gamma:\RxCat
  }{
    \Gamma.\II : \RxCat\\
    \StrMap{\Gamma.\II} = \StrMap{\Gamma}\times\II
  }
\end{mathpar}

Before defining the morphisms of $\RxCat$, we first characterize the term
variables:
\begin{mathpar}
  \inferrule[top variable]{
    \Mor[A]{\StrMap{\Gamma}}{\Tp}
  }{
    \IsVar{\Gamma.A}{\frk{z}_A}{A}
    \\
    \StrMap{\frk{z}_A} = q_A
  }
  \and
  \inferrule[pop variable]{
    \IsVar{\Gamma}{\alpha}{A}\\
    Q\in \brc{\II}\cup\brc{\Mor[B]{\StrMap{\Gamma}}{\Tp}}
  }{
    \IsVar{\Gamma.Q}{\frk{s}_Q\prn{\alpha}}{A\circ p_Q}
    \\
    \StrMap{\frk{s}_Q\prn{\alpha}} = \StrMap{\alpha}\circ p_Q
  }
\end{mathpar}

Then we define the morphisms $\Mor[\gamma]{\Delta}{\Gamma}$ simultaneously with
their decodings $\StrMap{\gamma}$ in terms of the cubical atomic terms:
\begin{mathpar}
  \inferrule[empty]{
  }{
    \Mor[\cdot]{\Delta}{\cdot}\\
    \StrMap{\cdot} = {!_{\StrMap{\Delta}}}
  }
  \and
  \inferrule[variable]{
    \Mor[\gamma]{\Delta}{\Gamma}\\
    \IsVar{\Delta}{\alpha}{A\circ\StrMap{\gamma}}
  }{
    \Mor[\gamma.\alpha]{\Delta}{\Gamma.A}\\
    \StrMap{\gamma.\alpha} = \gls{\StrMap{\gamma},\StrMap{\alpha}}
  }
  \and
  \inferrule[dimension]{
    \Mor[\gamma]{\Delta}{\Gamma}\\
    \Mor[r]{\StrMap{\Delta}}{\II}
  }{
    \Mor[\gamma.r]{\Delta}{\Gamma.\II}
    \\
    \StrMap{\gamma.r} = \gls{\StrMap{\gamma},r}
  }
  \and
\end{mathpar}

The above decodings assemble into a functor $\Mor[\StrMap]{\RxCat}{\ThCat}$,
which automatically induces an essential morphism of topoi that we will also
write $\Mor[\StrMap]{\RxTop}{\ThTop}$.

\begin{restatable}{lemma}{LemIntervalPreserved}\label{lem:ii-preserved}
  The chosen interval structure is preserved by restriction along
  $\Mor[\StrMap]{\RxTop}{\ThTop}$; that is, we have an isomorphism
  $\Yo[\RxCat]{\cdot.\II} \cong \StrMap^*\Yo[\ThCat]{\II}$.
\end{restatable}

\begin{construction}[The presheaf of variables]\label{con:var}
  We have a family $\Con{var}:\Sl{\Sh{\RxTop}}{\StrMap^*\Tm}$, whose fiber at
  each $\Mor[a]{\StrMap{\Gamma}}{\Tm\prn{A}}$ is the set of variables
  $\IsVar{\Gamma}{\frk{a}}{A}$ with $\StrMap{\frk{a}} = a$.
\end{construction}

\subsection{The glued topos}\label{sec:gluing}

Let $\SP$ be the \emph{Sierpi\'nski} topos satisfying $\Sh{\SP} =
{\SET}^{\to}$.  Writing $\PtTop$ for the punctual topos satisfying $\Sh{\PtTop*}
= \SET$, we have open and closed points $\Mor|open
immersion|[\circ]{\PtTop}{\SP}$ and $\Mor|closed
immersion|[\bullet]{\PtTop}{\SP}$ corresponding under inverse image to the
codomain and domain functors on $\SET$ respectively.
In geometrical terms, these endpoints render $\SP$ a \emph{directed interval}
$\brc{\bullet\to\circ}$ that can be used to form cylinders by cartesian
product.

We define the glued topos $\GlTop$ by gluing the open end of the cylinder
$\RxTop\times\SP$ onto $\ThTop$ along $\Mor[\StrMap]{\RxTop}{\ThTop}$, obtaining
a diagram of open and closed immersions $
  \tikz[diagram,baseline=(syn.base)]{
    \node (syn) [inner xsep=0pt] {$\ThTop\ $};
    \node (cmp) [right = 1.75cm of syn] {$\GlTop$};
    \node (sem) [right = 1.75cm of cmp,inner xsep=0pt] {$\ \RxTop$};
    \path[open immersion] (syn) edge node [desc] {$\OpImm$} (cmp);
    \path[closed immersion*] (sem) edge node [desc] {$\ClImm$} (cmp);
  }
$ as follows:
\[
  \begin{tikzpicture}[diagram]
    \SpliceDiagramSquare{
      se/style = pushout,
      west/style = open immersion,
      east/style = open immersion,
      nw = \RxTop,
      sw = \RxTop\times\SP,
      west = \RxTop\times{\circ},
      ne = \ThTop,
      north = \StrMap,
      se = \GlTop,
      east = \OpImm,
      west/node/style = upright desc,
    }
    \node (R') [left = 2.5cm of sw] {$\RxTop$};
    \path[closed immersion] (R') edge node [above] {$\RxTop\times\bullet$} (sw);
    \path[closed immersion, bend right=30] (R') edge node [sloped,below] {$\ClImm$} (se);
  \end{tikzpicture}
\]

Finally, we define the category $\CatSTC$ axiomatized in \cref{sec:stc} to be
$\CatSTC := \Sh{\GlTop}$.

\begin{remark}\label{rem:comma-cat}
  The gluing construction above has the advantage of being expressed totally in
  the category of Grothendieck topoi; for familiarity, we note that the
  category $\CatSTC = \Sh{\GlTop}$ is the traditional Artin gluing of the inverse image
  functor $\Mor[\StrMap^*]{\Psh{\ThCat}}{\Psh{\RxCat}}$ and can hence be
  explicitly computed as the comma category $\Psh{\RxCat}\downarrow \StrMap^*$.
  Therefore, an object of $\CatSTC$ is a pair $\prn{E, \Mor{A}{\StrMap^*E}}$
  of a presheaf $E : \Psh{\ThCat}$ and a presheaf $A : \Psh{\RxCat}$ equipped
  with a structure map to $\StrMap^*E$.
\end{remark}

\subsection{Verifying the axioms}\label{sec:axioms}

Because $\CatSTC = \Sh{\GlTop}$ is a category of sheaves, it satisfies Giraud's
axioms and thus interprets extensional type theory extended with a universe
$\Omega$ of propositions. To see that $\CatSTC$ moreover admits a hierarchy of
type-theoretic universes satisfying the strictification axiom
(\cref{ax:universes}), we observe that $\CatSTC$ also admits a presentation as a
presheaf category.

\begin{lemma}[Strong universes]\label{lem:gluing-psh}
  In $\CatSTC$ we have a cumulative and transfinite hierarchy of strong
  universes (\cref{def:op-universe}) $\UU_0\subseteq\UU_1\subseteq\ldots$
  corresponding to the sequence of strongly inaccessible cardinals in the
  background set theory.
\end{lemma}

\begin{proof}
  First, we observe that $\GlTop$ can be presented as $\PrTop{\CatIdent{G}}$ for
  some small category $\CatIdent{G}$, following a standard result of topos theory
  that the Artin gluing of a \emph{continuous} and accessible functor between
  categories of presheaves is again a category of
  presheaves~\cite{sga:4,carboni-johnstone:1995}. Note that $\StrMap^*$ is
  continuous, as the inverse image part of an essential morphism of topoi.

  We therefore obtain a hierarchy of Hofmann--Streicher universes
  \citep{hofmann-streicher:1997} in $\CatSTC\simeq\Psh{\CatIdent{G}}$, which
  Coquand has additionally shown to be cumulative. Finally, using the argument of
  Orton and Pitts~\citep[Theorem 6.3]{orton-pitts:2016} and our assumption
  that the background set theory is boolean, we conclude that the Hofmann--Streicher universes
  satisfy the strictification axiom.
\end{proof}

Using our computation of $\CatSTC$ as a comma category (\cref{rem:comma-cat}),
we can explicitly compute the proposition $\Syn$ stipulated by \cref{ax:syn}.
Then, we construct the syntactic $\ThCat$-model, family of term variables, and
interval object stipulated by Axioms~\ref{ax:syntactic-algebra}, \ref{ax:var},
and \ref{ax:interval} respectively.

\begin{computation}[The syntactic open]\label{cmp:syntactic-open}
  We define $\Syn : \CatSTC$ to be the subterminal object $\OpImm_!\ObjTerm =
  \prn{\ObjTerm,\Mor{\ObjInit}{\StrMap^*\ObjTerm}}$, for which we have an
  equivalence of categories $\Psh{\ThCat}\simeq\Sl{\CatSTC}{\Syn}$. The inverse
  image part of the open immersion can be viewed as the pullback map
  $\Mor[\Syn^*]{\CatSTC}{\Sl{\CatSTC}{\Syn}}$, and the direct image is the
  dependent product map $\Mor[\Syn_*]{\Sl{\CatSTC}{\Syn}}{\CatSTC}$.
\end{computation}

\begin{remark}[The syntactic $\ThCat$-model]\label{rem:syntactic-algebra}
  As the classifier of pre-models of $\ThCat$, the topos $\ThTop$ contains the
  \emph{generic pre-model} of $\ThCat$, namely the Yoneda embedding
  $\EmbMor{\ThCat}{\Sh{\ThTop} = \Psh{\ThCat}}$. Because the Yoneda embedding is
  locally Cartesian closed, the generic pre-model is in fact a genuine
  $\ThCat$-model; because $\ThCat$ is small, the size of this model is bounded
  by the smallest $\Op$-modal Hofmann--Streicher universe, allowing us to
  internalize it as an element of $\ThMod{\UUSyn}$.
\end{remark}

\begin{remark}[The family of term variables]\label{rem:glued-variables}
  In light of \cref{rem:comma-cat}, the presheaf
  $\Con{var}:\Sl{\Psh{\RxCat}}{\StrMap^*\Tm}$ from \cref{con:var} can be
  internalized as a $\UU$-valued family of types in $\CatSTC$ that restricts to
  $\Tm:\Psh{\ThCat}$ under the open immersion.
\end{remark}

\begin{construction}[The glued interval]\label{con:glued-interval}
  We take the interval $\II:\CatSTC$ to be the direct image $\II :=
  \OpImm_*\Yo[\ThCat]{\II}$ of the syntactic interval under the open immersion,
  so $\OpImm^*\II = \Yo[\ThCat]{\II}$. Because of the isomorphism from
  \cref{lem:ii-preserved}, it is easy enough to see that $\ClImm^*\II \cong
  \Yo[\RxCat]{\cdot.\II}$, completing our justification of
  \cref{ax:syntactic-algebra}.
\end{construction}

To see that the glued interval is tiny, we use a general fact about Artin
gluings along \emph{inverse image} functors.

\begin{restatable}{lemma}{LemGluingTinyObjects}\label{lem:gluing-tiny-objects}
  Let $\Mor[f]{\YTop}{\UTop}$ be a morphism of topoi. Write $\XTop$ for the
  Artin gluing of the inverse image functor $f^*$, and write $\Mor|open
  immersion|[j]{\UTop}{\XTop}$ and $\Mor|closed immersion|[i]{\YTop}{\XTop}$
  for the respective open and closed immersions of topoi.
  Suppose that $X:\Sh{\XTop}$ is a sheaf such that $j^*X$ is a tiny object in
  $\Sh{\UTop}$ and $i^*X$ is a tiny object in $\Sh{\YTop}$; then $X$ is tiny.
\end{restatable}
\begin{proof}
  It suffices to check that the exponential functor $\IHom*{X}{-}$ preserves
  colimits; see
  \ifarXiv
    the appendix for details.
  \else
    the appendix \citep{sterling-angiuli:2021:extended}.
  \fi
\end{proof}

\begin{corollary}\label{cor:interval-tiny}
  The glued interval is tiny in $\CatSTC$.
\end{corollary}

Finally, we verify \cref{ax:locality} and construct the normal and neutral forms
as indexed quotient inductive types \citep{altenkirch-kaposi:2016} valued in
$\UUSem$; the full definition appears in the
appendix~\ifarXiv{(\cref{fig:nf-ne-qit})}\else\citep{sterling-angiuli:2021:extended}\fi.

\begin{restatable}{lemma}{BotImpliesSyn}\label{lem:0=1-implies-syn}
  In the lattice of opens of $\GlTop$, we have $\FakeFalse\leq\Syn$.
\end{restatable}

\section{Normalization for cubical type theory}\label{sec:normalization-result}

Finally, we show that our computability model (\cref{thm:computability-algebra})
lets us compute the normal form of every syntactic type, implying the (external
to $\CatSTC$) decidability of type equality in cubical type theory, and the
injectivity of type constructors.

\begin{remark}
  Because $\Mor[\StrMap]{\RxTop}{\ThTop}$ is an essential morphism of topoi with
  additional left adjoint $\StrMap_!\dashv\StrMap^*$, so is the closed immersion
  $\Mor|closed immersion|[\ClImm]{\RxTop}{\GlTop}$; the additional left adjoint
  $\ClImm_!$ takes $E:\Psh{\RxCat}$ to the computability structure
  $\prn{\StrMap_!E,\Mor{E}{\StrMap^*\StrMap_!E}}$ determined by the unit of
  $\StrMap^*\StrMap_!$.
\end{remark}

\begin{construction}
  Let $\Gamma:\RxCat$ be an atomic context; we write $\Atom{\Gamma} : \CatSTC$
  for $\ClImm_!\Yo[\RxCat]{\Gamma}$, the computability structure of vectors of
  ``atoms of type $\Gamma$'' tracked by honest substitutions/terms.
\end{construction}

\begin{construction}
  The computability model evinces a locally Cartesian closed functor
  $\Mor[\CmpAlg]{\ThCat}{\CatSTC}$; restricting along
  $\Mor[\StrMap]{\RxCat}{\ThCat}$, we have an interpretation functor
  $\Mor[\bbrk{-}]{\RxCat}{\CatSTC}$ taking each atomic context to its
  computability structure $\CmpAlg\prn{\StrMap{\Gamma}}$. We observe that
  $\OpImm^*\bbrk{\Gamma} = \Yo[\ThCat]{\StrMap{\Gamma}} =
  \StrMap_!\Yo[\RxCat]{\Gamma}$.
\end{construction}

\begin{construction}
  For any $X:\CatSTC$ there is a canonical natural transformation
  $\Mor{\Hom*{\bbrk{-}}{X}}{\StrMap^*\OpImm^*X} : \Psh{\RxCat}$ which restricts
  $\Mor{\CmpAlg\prn{\StrMap{\Gamma}}}{X}:\CatSTC$ to its syntactic part, noting
  that $\Hom{\StrMap_!\Yo[\RxCat]{\Gamma}}{\OpImm^*X} \cong
  \StrMap^*\OpImm^*X\prn{\Gamma}$. Viewed as a sheaf on $\GlTop$, we write
  $X_{\CmpAlg} : \CatSTC$ for the pair $\prn{\OpImm^*X,
  \Mor{\Hom*{\bbrk{-}}{X}}{\StrMap^*\OpImm^*X}}$.
\end{construction}

\begin{construction}\label{con:atm-to-cx}
  We define a pointwise vertical (\cref{rem:vertical}) natural transformation
  $\Mor[\AtmToCx]{\Atom}{\bbrk{-}} : \Hom{\RxCat}{\CatSTC}$ that reflects
  each atomic substitution as a computable substitution. The definition follows
  by recursion on the index $\Gamma:\RxCat$, and uses the fact that the locus
  of instability of a variable is empty:
  \begin{align*}
    \AtmToCx[\prn{\cdot}]{\cdot} &= \cdot\\
    \AtmToCx[\Gamma.\II]{\gamma.r} &= \prn{\AtmToCx[\Gamma]{\gamma}, r}\\
    \AtmToCx[\Gamma.A]{\gamma.x} &= \prn{\AtmToCx[\Gamma]{\gamma}, \Reflect{\CmpAlg\prn{A}}{\bot_\FF}{\NeGlue{\frk{var}\prn{x}}{\bot_\FF}{\brk{}\Sub{\CmpAlg\prn{A}}}}}
  \end{align*}
\end{construction}

\begin{restatable}{lemma}{ReixAtom}\label{lem:reix-atom}
  The pointwise vertical natural transformation
  $\Mor[\AtmToCx]{\Atom}{\bbrk{-}}$ induces by precomposition a vertical map
  $\Mor[\AtmToCx^*_X]{X\Sub{\CmpAlg}}{X}$ for any sheaf $X:\CatSTC$.
\end{restatable}

\begin{restatable}[The normalization function]{theorem}{NormFunction}\label{thm:norm-fn}
  The functor $\Mor[\CmpAlg]{\ThCat}{\CatSTC}$ induces a vertical map
  $\Mor{\SynAlg.\Tp}{\prn{\CmpAlg.\Tp}\Sub{\CmpAlg}}$. Composing this with the
  vertical maps $\AtmToCx^*\Sub{\CmpAlg.\Tp}$ and $\NormTp$, we obtain a
  vertical \emph{normalization} map $\Con{nbe}$ sending a syntactic type to the
  normal form chosen by its normalization structure in $\CmpAlg$:
  \[
    \begin{tikzpicture}[diagram,baseline=(0.base)]
      \node (0) {$\SynAlg.\Tp$};
      \node (1) [right = of 0] {$\prn{\CmpAlg.\Tp}\Sub{\CmpAlg}$};
      \node (2) [right = 3.2cm of 1] {$\CmpAlg.\Tp$};
      \node (3) [right = of 2] {$\NfTp$};
      \path[->,exists] (0) edge (1);
      \path[->] (1) edge node [above] {$\AtmToCx^*\Sub{\CmpAlg.\Tp}$} (2);
      \path[->] (2) edge node [above] {$\NormTp$} (3);
    \end{tikzpicture}
  \]
\end{restatable}

We can similarly exhibit a pointwise vertical normalization function for
syntactic \emph{terms}:
\[
  \Prod{A:\SynAlg.\Tp}{\Ext{\IHom{\SynAlg.\Tm\prn{A}}{\Nf{A}}}{\Syn}{\ArrId{}}}
\]

The standard correctness conditions (soundness and completeness) for
normalization follow immediately.

\begin{theorem}[Correctness of normalization]\label{thm:soundness}
  The normalization function is sound and complete for cubical type theory.
  \begin{enumerate}
    \item \emph{Completeness} --- if two (types, terms) are equal, then they are taken to equal normal forms.
    \item \emph{Soundness} --- if two (types, terms) are taken to the same normal form, then they are equal.
  \end{enumerate}
\end{theorem}

\begin{proof}
  Completeness is automatic because our entire development was carried out
  relative to judgmental equivalence classes of terms. Soundness follows from
  the fact that the normalization function is vertical, hence a section to the
  unit of the open modality, hence a monomorphism.
\end{proof}

We would not expect the $\Pi$ constructor to be injective in the syntactic
category (a derivability): because monomorphisms are preserved by left exact
functors, this would imply that any model of cubical type theory has injective
type constructors. However, there is a \emph{modal} version of injectivity
corresponding to the traditional admissibility statement that \emph{does} hold.

\begin{restatable}[Injectivity of type constructors]{theorem}{Injectivity}\label{thm:injectivity}
  The following formula holds in the internal logic of $\CatSTC = \Sh{\GlTop}$:
  \begin{multline*}
    \forall A,A',B,B'.\\
    \SynAlg.\Pi\prn{A,B} = \SynAlg.\Pi\prn{A',B'}
    \implies
    \Cl\prn{
      \prn{A,B} = \prn{A',B'}
    }
  \end{multline*}
\end{restatable}

\begin{lemma}[Idempotence of normalization]\label{lem:idemp}
  For any normal type $A$ we have $\Con{nbe}\prn{A} = A : \NfTp$; and likewise,
  for any normal term $a$ we have $\Con{nbe}\prn{a} = a : \Nf{A}$.
\end{lemma}

\begin{corollary}\label{cor:tightness}
  The normal form presentation is \emph{tight}: a given term has only a single
  normal form. Hence the normalization function is an isomorphism.
\end{corollary}

\begin{proof}
  Suppose we have distinct normal forms $\frk{a} \not=\frk{a}'$ for a term $a$;
  the normalization function is vertical, so at least one of $\frk{a},\frk{a}'$
  lies outside its image; but this contradicts \cref{lem:idemp}.
\end{proof}

Both the syntax and the normal forms of cubical type theory are recursively
enumerable, because they can be presented by finitely many rules in a
conventional deductive system. Hence we have the following corollaries of
\cref{thm:norm-fn,thm:soundness,cor:tightness}.

\begin{corollary}[Normalization algorithm]\label{thm:norm-alg}
  The normalization function (\cref{thm:norm-fn}) is tracked by a recursive
  function sending well-typed raw (types, terms) to normal forms.
\end{corollary}

\begin{corollary}[Decidability of equality]\label{thm:dec-eq}
  Judgmental equality of types and terms in atomic contexts is decidable.
\end{corollary}

\begin{proof}
  First we obtain their unique normal forms using the algorithm described in
  \cref{thm:norm-alg}; because equality of normal forms is (externally) decidable, we are done.
\end{proof}

\begin{remark}
  The above results unfold to statements about judgments in atomic contexts
  (i.e., in the image of $\Mor[\StrMap]{\RxCat}{\ThCat}$), but standard
  presentations of cubical type theory also allow context extension by
  $\_:\brk{\phi}$. However, one can algorithmically eliminate such assumptions
  by left inversion as in the implementations of Cubical Agda and
  \redtt{}.
\end{remark}
 
\section*{Acknowledgments}

We thank Mathieu Anel, Steve Awodey, Rafa\"el Bocquet, Thierry Coquand,
Favonia, Marcelo Fiore, Daniel Gratzer, Robert Harper, Kenji Maillard, Anders
M\"ortberg, Michael Shulman, Thomas Streicher, and Andrea Vezzosi for useful
conversations on cubical computability and Artin gluing, and for pointing out
mistakes in previous drafts of this paper. We thank Tristan Nguyen at AFOSR for
his support.

This work was supported in part by AFOSR under MURI grants FA9550-15-1-0053,
FA9550-19-1-0216, and FA9550-21-0009. Any opinions, findings and conclusions or
recommendations expressed in this material are those of the authors and do not
necessarily reflect the views of the AFOSR.

\nocite{coquand:2019}
\nocite{sterling-gratzer:2020:stc,sterling:thesis-proposal,sterling-harper:2020}
\nocite{anel-joyal:2021,vickers:2007}
\nocite{rijke-shulman-spitters:2017,shulman:blog:internalizing}

\nocite{bmss:2011,kock:2006,kock:2009,bauer:2006,bunge-gago-san-luis:2018,hyland:1991,blechschmidt:2017,orton-pitts:2016,hottbook}

\bibliographystyle{IEEEtranSN}
\bibliography{refs-local,temp-refs}

% Generated by IEEEtranSN.bst, version: 1.14 (2015/08/26)
\begin{thebibliography}{68}
\providecommand{\natexlab}[1]{#1}
\providecommand{\url}[1]{#1}
\csname url@samestyle\endcsname
\providecommand{\newblock}{\relax}
\providecommand{\bibinfo}[2]{#2}
\providecommand{\BIBentrySTDinterwordspacing}{\spaceskip=0pt\relax}
\providecommand{\BIBentryALTinterwordstretchfactor}{4}
\providecommand{\BIBentryALTinterwordspacing}{\spaceskip=\fontdimen2\font plus
\BIBentryALTinterwordstretchfactor\fontdimen3\font minus
  \fontdimen4\font\relax}
\providecommand{\BIBforeignlanguage}[2]{{%
\expandafter\ifx\csname l@#1\endcsname\relax
\typeout{** WARNING: IEEEtranSN.bst: No hyphenation pattern has been}%
\typeout{** loaded for the language `#1'. Using the pattern for}%
\typeout{** the default language instead.}%
\else
\language=\csname l@#1\endcsname
\fi
#2}}
\providecommand{\BIBdecl}{\relax}
\BIBdecl

\bibitem[Abel(2013)]{abel:2013}
A.~Abel, ``Normalization by evaluation: Dependent types and impredicativity,''
  Habilitation, Ludwig-Maximilians-Universit\"{a}t M\"{u}nchen, 2013.

\bibitem[Altenkirch et~al.(2001)Altenkirch, Dybjer, Hofmann, and
  Scott]{altenkirch-dybjer-hofmann-scott:2001}
T.~Altenkirch, P.~Dybjer, M.~Hofmann, and P.~Scott, ``Normalization by
  evaluation for typed lambda calculus with coproducts,'' in \emph{{Proceedings
  of the 16th Annual IEEE Symposium on Logic in Computer Science}}.\hskip 1em
  plus 0.5em minus 0.4em\relax Washington, DC, USA: IEEE Computer Society,
  2001, pp. 303--.

\bibitem[Altenkirch and Kaposi(2016)]{altenkirch-kaposi:2016}
T.~Altenkirch and A.~Kaposi, ``Type theory in type theory using quotient
  inductive types,'' in \emph{Proceedings of the 43rd Annual ACM SIGPLAN-SIGACT
  Symposium on Principles of Programming Languages}, ser. POPL '16.\hskip 1em
  plus 0.5em minus 0.4em\relax St. Petersburg, FL, USA: ACM, 2016, pp. 18--29.

\bibitem[Altenkirch et~al.(1995)Altenkirch, Hofmann, and
  Streicher]{altenkirch-hofmann-streicher:1995}
T.~Altenkirch, M.~Hofmann, and T.~Streicher, ``Categorical reconstruction of a
  reduction free normalization proof,'' in \emph{Category Theory and Computer
  Science}, D.~Pitt, D.~E. Rydeheard, and P.~Johnstone, Eds.\hskip 1em plus
  0.5em minus 0.4em\relax Berlin, Heidelberg: Springer Berlin Heidelberg, 1995,
  pp. 182--199.

\bibitem[Anel and Joyal(2021)]{anel-joyal:2021}
M.~Anel and A.~Joyal, ``Topo-logie,'' in \emph{New Spaces in Mathematics:
  Formal and Conceptual Reflections}, M.~Anel and G.~Catren, Eds.\hskip 1em
  plus 0.5em minus 0.4em\relax Cambridge University Press, 2021, vol.~1, ch.~4,
  pp. 155--257.

\bibitem[Angiuli(2019)]{angiuli:2019}
C.~Angiuli, ``Computational semantics of {Cartesian} cubical type theory,''
  Ph.D. dissertation, Carnegie Mellon University, 2019.

\bibitem[Angiuli et~al.(2017)Angiuli, Harper, and
  Wilson]{angiuli-harper-wilson:2017}
C.~Angiuli, R.~Harper, and T.~Wilson, ``Computational higher-dimensional type
  theory,'' in \emph{POPL 2017: Proceedings of the 44th ACM SIGPLAN Symposium
  on Principles of Programming Languages}.\hskip 1em plus 0.5em minus
  0.4em\relax Paris, France: ACM, 2017, pp. 680--693.

\bibitem[Angiuli et~al.(2018)Angiuli, {Hou (Favonia)}, and
  Harper]{angiuli-favonia-harper:2018}
\BIBentryALTinterwordspacing
C.~Angiuli, K.-B. {Hou (Favonia)}, and R.~Harper, ``{Cartesian Cubical
  Computational Type Theory: Constructive Reasoning with Paths and
  Equalities},'' in \emph{27th EACSL Annual Conference on Computer Science
  Logic (CSL 2018)}, ser. Leibniz International Proceedings in Informatics
  (LIPIcs), D.~Ghica and A.~Jung, Eds., vol. 119.\hskip 1em plus 0.5em minus
  0.4em\relax Dagstuhl, Germany: Schloss Dagstuhl--Leibniz-Zentrum fuer
  Informatik, 2018, pp. 6:1--6:17. [Online]. Available:
  \url{http://drops.dagstuhl.de/opus/volltexte/2018/9673}
\BIBentrySTDinterwordspacing

\bibitem[Angiuli et~al.(2019)Angiuli, Brunerie, Coquand, {Hou (Favonia)},
  Harper, and Licata]{abcfhl:2019}
\BIBentryALTinterwordspacing
C.~Angiuli, G.~Brunerie, T.~Coquand, K.-B. {Hou (Favonia)}, R.~Harper, and
  D.~R. Licata, ``Syntax and models of {Cartesian} cubical type theory,'' Feb.
  2019, preprint. [Online]. Available:
  \url{https://github.com/dlicata335/cart-cube}
\BIBentrySTDinterwordspacing

\bibitem[Artin et~al.(1972)Artin, Grothendieck, and Verdier]{sga:4}
M.~Artin, A.~Grothendieck, and J.-L. Verdier, \emph{Th\'{e}orie des topos et
  cohomologie \'{e}tale des sch\'{e}mas}.\hskip 1em plus 0.5em minus
  0.4em\relax Berlin: Springer-Verlag, 1972, s\'{e}minaire de G\'{e}om\'{e}trie
  Alg\'{e}brique du Bois-Marie 1963--1964 (SGA 4), Dirig\'{e} par M. Artin, A.
  Grothendieck, et J.-L. Verdier. Avec la collaboration de N. Bourbaki, P.
  Deligne et B. Saint-Donat, Lecture Notes in Mathematics, Vol. 269, 270, 305.

\bibitem[Awodey and Butz(2000)]{awodey-butz:2000}
\BIBentryALTinterwordspacing
S.~Awodey and C.~Butz, ``Topological completeness for higher-order logic,''
  \emph{The Journal of Symbolic Logic}, vol.~65, no.~3, pp. 1168--1182, 2000.
  [Online]. Available: \url{http://www.jstor.org/stable/2586693}
\BIBentrySTDinterwordspacing

\bibitem[Awodey(2018)]{awodey:2018:natural-models}
S.~Awodey, ``Natural models of homotopy type theory,'' \emph{Mathematical
  Structures in Computer Science}, vol.~28, no.~2, pp. 241--286, 2018.

\bibitem[Bauer(2006)]{bauer:2006}
A.~Bauer, ``First steps in synthetic computability theory,'' \emph{Electronic
  Notes in Theoretical Computer Science}, vol. 155, pp. 5--31, 2006,
  proceedings of the 21st Annual Conference on Mathematical Foundations of
  Programming Semantics (MFPS XXI).

\bibitem[Birkedal et~al.(2011)Birkedal, M{\o{}}gelberg, Schwinghammer, and
  Stovring]{bmss:2011}
L.~Birkedal, R.~E. M{\o{}}gelberg, J.~Schwinghammer, and K.~Stovring, ``First
  steps in synthetic guarded domain theory: Step-indexing in the topos of
  trees,'' in \emph{Proceedings of the 2011 IEEE 26th Annual Symposium on Logic
  in Computer Science}.\hskip 1em plus 0.5em minus 0.4em\relax Washington, DC,
  USA: IEEE Computer Society, 2011, pp. 55--64.

\bibitem[Birkedal et~al.(2016)Birkedal, Bizjak, Clouston, Grathwohl, Spitters,
  and Vezzosi]{bbcgsv:2016}
L.~Birkedal, A.~Bizjak, R.~Clouston, H.~B. Grathwohl, B.~Spitters, and
  A.~Vezzosi, ``{Guarded Cubical Type Theory: Path Equality for Guarded
  Recursion},'' in \emph{25th EACSL Annual Conference on Computer Science Logic
  (CSL 2016)}, ser. Leibniz International Proceedings in Informatics (LIPIcs),
  J.-M. Talbot and L.~Regnier, Eds., vol.~62.\hskip 1em plus 0.5em minus
  0.4em\relax Dagstuhl, Germany: Schloss Dagstuhl--Leibniz-Zentrum fuer
  Informatik, 2016, pp. 23:1--23:17.

\bibitem[Blechschmidt(2017)]{blechschmidt:2017}
I.~Blechschmidt, ``Using the internal language of toposes in algebraic
  geometry,'' Ph.D. dissertation, Universit\"{a}t Augsberg, 2017.

\bibitem[Bunge and Funk(2006)]{bunge-funk:2006}
M.~Bunge and J.~Funk, \emph{Singular coverings of toposes}, ser. Lecture notes
  in mathematics, 1890.\hskip 1em plus 0.5em minus 0.4em\relax Berlin:
  Springer, 2006.

\bibitem[Bunge et~al.(2018)Bunge, Gago, and San~Luis]{bunge-gago-san-luis:2018}
M.~Bunge, F.~Gago, and A.~M. San~Luis, \emph{Synthetic Differential Topology},
  ser. London Mathematical Society Lecture Note Series.\hskip 1em plus 0.5em
  minus 0.4em\relax Cambridge University Press, 2018.

\bibitem[Carboni and Johnstone(1995)]{carboni-johnstone:1995}
A.~Carboni and P.~Johnstone, ``Connected limits, familial representability and
  {Artin} glueing,'' \emph{Mathematical Structures in Computer Science},
  vol.~5, no.~4, pp. 441--459, 1995.

\bibitem[Cohen et~al.(2017)Cohen, Coquand, Huber, and M\"{o}rtberg]{cchm:2017}
\BIBentryALTinterwordspacing
C.~Cohen, T.~Coquand, S.~Huber, and A.~M\"{o}rtberg, ``{Cubical Type Theory: a
  constructive interpretation of the univalence axiom},'' \emph{IfCoLog Journal
  of Logics and their Applications}, vol.~4, no.~10, pp. 3127--3169, Nov. 2017.
  [Online]. Available:
  \url{http://www.collegepublications.co.uk/journals/ifcolog/?00019}
\BIBentrySTDinterwordspacing

\bibitem[Coquand(2019)]{coquand:2019}
T.~Coquand, ``Canonicity and normalization for dependent type theory,''
  \emph{Theoretical Computer Science}, vol. 777, pp. 184--191, 2019, in memory
  of Maurice Nivat, a founding father of Theoretical Computer Science - Part I.

\bibitem[Coquand et~al.(2018)Coquand, Huber, and
  M\"{o}rtberg]{coquand-huber-mortberg:2018}
T.~Coquand, S.~Huber, and A.~M\"{o}rtberg, ``On higher inductive types in
  cubical type theory,'' in \emph{Proceedings of the 33rd Annual ACM/IEEE
  Symposium on Logic in Computer Science}.\hskip 1em plus 0.5em minus
  0.4em\relax Oxford, United Kingdom: ACM, 2018, pp. 255--264.

\bibitem[Coquand et~al.(2019)Coquand, Huber, and
  Sattler]{coquand-huber-sattler:2019}
T.~Coquand, S.~Huber, and C.~Sattler, ``Homotopy canonicity for cubical type
  theory,'' in \emph{4th International Conference on Formal Structures for
  Computation and Deduction (FSCD 2019)}, ser. Leibniz International
  Proceedings in Informatics (LIPIcs), H.~Geuvers, Ed., vol. 131.\hskip 1em
  plus 0.5em minus 0.4em\relax Dagstuhl, Germany: Schloss
  Dagstuhl--Leibniz-Zentrum fuer Informatik, 2019.

\bibitem[Crole(1993)]{crole:1993}
R.~L. Crole, \emph{Categories for Types}, ser. Cambridge Mathematical
  Textbooks.\hskip 1em plus 0.5em minus 0.4em\relax New York: Cambridge
  University Press, 1993.

\bibitem[Diaconescu(1975)]{diaconescu:1975}
R.~Diaconescu, ``Change of base for toposes with generators,'' \emph{Journal of
  Pure and Applied Algebra}, vol.~6, no.~3, pp. 191--218, 1975.

\bibitem[Dybjer(1996)]{dybjer:1996}
P.~Dybjer, ``Internal type theory,'' in \emph{Types for Proofs and Programs:
  International Workshop, TYPES '95 Torino, Italy, June 5--8, 1995 Selected
  Papers}, S.~Berardi and M.~Coppo, Eds.\hskip 1em plus 0.5em minus 0.4em\relax
  Berlin, Heidelberg: Springer Berlin Heidelberg, 1996, pp. 120--134.

\bibitem[Fiore(2002)]{fiore:2002}
M.~Fiore, ``Semantic analysis of normalisation by evaluation for typed lambda
  calculus,'' in \emph{Proceedings of the 4th ACM SIGPLAN International
  Conference on Principles and Practice of Declarative Programming}, ser. PPDP
  '02.\hskip 1em plus 0.5em minus 0.4em\relax Pittsburgh, PA, USA: ACM, 2002,
  pp. 26--37.

\bibitem[Freyd(1978)]{freyd:1978}
P.~Freyd, ``On proving that $\mathbf{1}$ is an indecomposable projective in
  various free categories,'' 1978, unpublished manuscript.

\bibitem[Gratzer and Sterling(2020)]{gratzer-sterling:2020}
\BIBentryALTinterwordspacing
D.~Gratzer and J.~Sterling, ``Syntactic categories for dependent type theory:
  sketching and adequacy,'' 2020, preprint. [Online]. Available:
  \url{https://arxiv.org/abs/2012.10783}
\BIBentrySTDinterwordspacing

\bibitem[Gratzer et~al.(2020)Gratzer, Kavvos, Nuyts, and
  Birkedal]{gratzer-kavvos-nuyts-birkedal:2020}
D.~Gratzer, G.~A. Kavvos, A.~Nuyts, and L.~Birkedal, ``Multimodal dependent
  type theory,'' in \emph{Proceedings of the 35th Annual ACM/IEEE Symposium on
  Logic in Computer Science}.\hskip 1em plus 0.5em minus 0.4em\relax
  Saarbr\"{u}cken, Germany: Association for Computing Machinery, 2020, pp.
  492--506.

\bibitem[Harper and Licata(2007)]{harper-licata:2007}
R.~Harper and D.~R. Licata, ``Mechanizing metatheory in a logical framework,''
  \emph{Journal of Functional Programming}, vol.~17, no. 4-5, pp. 613--673,
  Jul. 2007.

\bibitem[Harper et~al.(1993)Harper, Honsell, and
  Plotkin]{harper-honsell-plotkin:1993}
R.~Harper, F.~Honsell, and G.~Plotkin, ``A framework for defining logics,''
  \emph{J. ACM}, vol.~40, no.~1, pp. 143--184, Jan. 1993.

\bibitem[Henkin(1950)]{henkin:1950}
\BIBentryALTinterwordspacing
L.~Henkin, ``Completeness in the theory of types,'' \emph{The Journal of
  Symbolic Logic}, vol.~15, no.~2, pp. 81--91, 1950. [Online]. Available:
  \url{http://www.jstor.org/stable/2266967}
\BIBentrySTDinterwordspacing

\bibitem[Hofmann and Streicher(1997)]{hofmann-streicher:1997}
\BIBentryALTinterwordspacing
M.~Hofmann and T.~Streicher, ``Lifting {G}rothendieck universes,'' 1997,
  unpublished note. [Online]. Available:
  \url{https://www2.mathematik.tu-darmstadt.de/~streicher/NOTES/lift.pdf}
\BIBentrySTDinterwordspacing

\bibitem[Huber(2018)]{huber:2018}
S.~Huber, ``Canonicity for cubical type theory,'' \emph{Journal of Automated
  Reasoning}, Jun. 2018.

\bibitem[Hyland(1991)]{hyland:1991}
J.~M.~E. Hyland, ``First steps in synthetic domain theory,'' in \emph{Category
  Theory}, A.~Carboni, M.~C. Pedicchio, and G.~Rosolini, Eds.\hskip 1em plus
  0.5em minus 0.4em\relax Berlin, Heidelberg: Springer Berlin Heidelberg, 1991,
  pp. 131--156.

\bibitem[Johnstone(2002)]{johnstone:2002}
P.~T. Johnstone, \emph{Sketches of an Elephant: A Topos Theory Compendium:
  Volumes 1 and 2}, ser. Oxford Logical Guides.\hskip 1em plus 0.5em minus
  0.4em\relax Oxford Science Publications, 2002, no.~43.

\bibitem[Jung and Tiuryn(1993)]{jung-tiuryn:1993}
A.~Jung and J.~Tiuryn, ``A new characterization of lambda definability,'' in
  \emph{Typed Lambda Calculi and Applications}, M.~Bezem and J.~F. Groote,
  Eds.\hskip 1em plus 0.5em minus 0.4em\relax Berlin, Heidelberg: Springer
  Berlin Heidelberg, 1993, pp. 245--257.

\bibitem[Kaposi(2017)]{kaposi:thesis}
A.~Kaposi, ``Type theory in a type theory with quotient inductive types,''
  Ph.D. dissertation, University of Nottingham, 2017.

\bibitem[Kaposi et~al.(2019)Kaposi, Huber, and
  Sattler]{kaposi-huber-sattler:2019}
A.~Kaposi, S.~Huber, and C.~Sattler, ``Gluing for type theory,'' in \emph{4th
  International Conference on Formal Structures for Computation and Deduction
  (FSCD 2019)}, ser. Leibniz International Proceedings in Informatics (LIPIcs),
  H.~Geuvers, Ed., vol. 131.\hskip 1em plus 0.5em minus 0.4em\relax Dagstuhl,
  Germany: Schloss Dagstuhl--Leibniz-Zentrum fuer Informatik, 2019.

\bibitem[Kapulkin and Sattler(2019)]{kapulkin-sattler:2019}
\BIBentryALTinterwordspacing
C.~Kapulkin and C.~Sattler, ``Homotopy canonicity of homotopy type theory,''
  Aug. 2019, slides from a talk given at the International Conference on
  Homotopy Type Theory (HoTT 2019). [Online]. Available:
  \url{https://hott.github.io/HoTT-2019/conf-slides/Sattler.pdf}
\BIBentrySTDinterwordspacing

\bibitem[Kock(2006)]{kock:2006}
A.~Kock, \emph{Synthetic Differential Geometry}, 2nd~ed., ser. London
  Mathematical Society Lecture Note Series.\hskip 1em plus 0.5em minus
  0.4em\relax Cambridge University Press, 2006.

\bibitem[Kock(2009)]{kock:2009}
------, \emph{Synthetic Geometry of Manifolds}, ser. Cambridge Tracts in
  Mathematics.\hskip 1em plus 0.5em minus 0.4em\relax Cambridge University
  Press, 2009.

\bibitem[Lawvere(1980)]{lawvere:1980}
F.~W. Lawvere, ``Toward the description in a smooth topos of the dynamically
  possible motions and deformations of a continuous body,'' \emph{Cahiers de
  Topologie et G\'{e}om\'{e}trie Diff\'{e}rentielle Cat\'{e}goriques}, vol.~21,
  no.~4, pp. 377--392, 1980.

\bibitem[Licata et~al.(2018)Licata, Orton, Pitts, and Spitters]{lops:2018}
D.~R. Licata, I.~Orton, A.~M. Pitts, and B.~Spitters, ``Internal universes in
  models of homotopy type theory,'' in \emph{3rd International Conference on
  Formal Structures for Computation and Deduction, {FSCD} 2018, July 9-12,
  2018, Oxford, {UK}}, 2018, pp. 22:1--22:17.

\bibitem[Newstead(2018)]{newstead:2018}
C.~Newstead, ``Algebraic models of dependent type theory,'' Ph.D. dissertation,
  Carnegie Mellon University, 2018.

\bibitem[Nordstr\"{o}m et~al.(1990)Nordstr\"{o}m, Peterson, and
  Smith]{nordstrom-peterson-smith:1990}
B.~Nordstr\"{o}m, K.~Peterson, and J.~M. Smith, \emph{Programming in
  Martin-L\"{o}f's Type Theory}, ser. International Series of Monographs on
  Computer Science.\hskip 1em plus 0.5em minus 0.4em\relax NY: Oxford
  University Press, 1990, vol.~7.

\bibitem[Orton and Pitts(2016)]{orton-pitts:2016}
I.~Orton and A.~M. Pitts, ``Axioms for modelling cubical type theory in a
  topos,'' in \emph{25th {EACSL} Annual Conference on Computer Science Logic,
  {CSL} 2016, August 29 - September 1, 2016, Marseille, France}, 2016, pp.
  24:1--24:19.

\bibitem[Pfenning and Sch\"{u}rmann(1999)]{pfenning-schuermann:1999}
F.~Pfenning and C.~Sch\"{u}rmann, ``System description: Twelf --- a
  meta-logical framework for deductive systems,'' in \emph{Automated Deduction
  --- CADE-16}.\hskip 1em plus 0.5em minus 0.4em\relax Berlin, Heidelberg:
  Springer Berlin Heidelberg, 1999, pp. 202--206.

\bibitem[{RedPRL Development Team}(2018)]{redtt:2018}
\BIBentryALTinterwordspacing
{RedPRL Development Team}, ``{\redtt},'' 2018. [Online]. Available:
  \url{https://www.github.com/RedPRL/redtt}
\BIBentrySTDinterwordspacing

\bibitem[Riehl and Shulman(2017)]{riehl-shulman:2017}
E.~Riehl and M.~Shulman, ``A type theory for synthetic $\infty$-categories,''
  \emph{Higher Structures}, vol.~1, 2017.

\bibitem[Rijke et~al.(2020)Rijke, Shulman, and
  Spitters]{rijke-shulman-spitters:2017}
\BIBentryALTinterwordspacing
E.~Rijke, M.~Shulman, and B.~Spitters, ``Modalities in homotopy type theory,''
  \emph{{Logical Methods in Computer Science}}, vol. {Volume 16, Issue 1}, Jan.
  2020. [Online]. Available: \url{https://lmcs.episciences.org/6015}
\BIBentrySTDinterwordspacing

\bibitem[Shulman(2015)]{shulman:2015}
M.~Shulman, ``Univalence for inverse diagrams and homotopy canonicity,''
  \emph{Mathematical Structures in Computer Science}, vol.~25, no.~5, pp.
  1203--1277, 2015.

\bibitem[Shulman(2011)]{shulman:blog:internalizing}
\BIBentryALTinterwordspacing
------, ``Internalizing the external, or the joys of codiscreteness,'' 2011.
  [Online]. Available:
  \url{https://golem.ph.utexas.edu/category/2011/11/internalizing_the_external_or.html}
\BIBentrySTDinterwordspacing

\bibitem[Sterling(2020)]{sterling:thesis-proposal}
\BIBentryALTinterwordspacing
J.~Sterling, ``{Objective Metatheory of (Cubical) Type Theories},'' 2020,
  {Thesis Proposal}. [Online]. Available:
  \url{http://www.jonmsterling.com/pdfs/proposal-slides.pdf}
\BIBentrySTDinterwordspacing

\bibitem[Sterling and Gratzer(2020)]{sterling-gratzer:2020:stc}
J.~Sterling and D.~Gratzer, ``Lectures on {Synthetic Tait Computability},''
  2020, notes on a lecture given by Sterling in Summer 2020.

\bibitem[Sterling and Harper(2020)]{sterling-harper:2020}
\BIBentryALTinterwordspacing
J.~Sterling and R.~Harper, ``{Logical Relations As Types: Proof-Relevant
  Parametricity for Program Modules},'' 2020, under review. [Online].
  Available: \url{https://arxiv.org/abs/2010.08599}
\BIBentrySTDinterwordspacing

\bibitem[Sterling et~al.(2019)Sterling, Angiuli, and
  Gratzer]{sterling-angiuli-gratzer:2019}
\BIBentryALTinterwordspacing
J.~Sterling, C.~Angiuli, and D.~Gratzer, ``Cubical syntax for reflection-free
  extensional equality,'' in \emph{4th International Conference on Formal
  Structures for Computation and Deduction (FSCD 2019)}, ser. Leibniz
  International Proceedings in Informatics (LIPIcs), H.~Geuvers, Ed., vol.
  131.\hskip 1em plus 0.5em minus 0.4em\relax Dagstuhl, Germany: Schloss
  Dagstuhl--Leibniz-Zentrum fuer Informatik, 2019, pp. 31:1--31:25. [Online].
  Available: \url{http://drops.dagstuhl.de/opus/volltexte/2019/10538}
\BIBentrySTDinterwordspacing

\bibitem[Sterling et~al.(2020)Sterling, Angiuli, and
  Gratzer]{sterling-angiuli-gratzer:2020}
\BIBentryALTinterwordspacing
------, ``A cubical language for {Bishop} sets,'' 2020, under review. [Online].
  Available: \url{https://arxiv.org/abs/2003.01491}
\BIBentrySTDinterwordspacing

\bibitem[Streicher(1998)]{streicher:1998}
T.~Streicher, ``Categorical intuitions underlying semantic normalisation
  proofs,'' in \emph{Preliminary Proceedings of the {APPSEM} Workshop on
  Normalisation by Evaluation}, O.~Danvy and P.~Dybjer, Eds.\hskip 1em plus
  0.5em minus 0.4em\relax Department of Computer Science, Aarhus University,
  1998.

\bibitem[Tait(1967)]{tait:1967}
\BIBentryALTinterwordspacing
W.~W. Tait, ``{Intensional Interpretations of Functionals of Finite Type I},''
  \emph{The Journal of Symbolic Logic}, vol.~32, no.~2, pp. 198--212, 1967.
  [Online]. Available: \url{http://www.jstor.org/stable/2271658}
\BIBentrySTDinterwordspacing

\bibitem[Taylor(1999)]{taylor:1999}
P.~Taylor, \emph{Practical Foundations of Mathematics}, ser. Cambridge studies
  in advanced mathematics.\hskip 1em plus 0.5em minus 0.4em\relax Cambridge,
  New York (N. Y.), Melbourne: Cambridge University Press, 1999.

\bibitem[Uemura(2017)]{uemura:2017}
\BIBentryALTinterwordspacing
T.~Uemura, ``Fibred fibration categories,'' in \emph{Proceedings of the 32nd
  Annual ACM/IEEE Symposium on Logic in Computer Science}.\hskip 1em plus 0.5em
  minus 0.4em\relax Reykjavik, Iceland: IEEE Press, 2017, pp. 24:1--24:12.
  [Online]. Available: \url{http://dl.acm.org/citation.cfm?id=3329995.3330019}
\BIBentrySTDinterwordspacing

\bibitem[Uemura(2019)]{uemura:2019}
\BIBentryALTinterwordspacing
------, ``A general framework for the semantics of type theory,'' 2019,
  preprint. [Online]. Available: \url{https://arxiv.org/abs/1904.04097}
\BIBentrySTDinterwordspacing

\bibitem[{Univalent Foundations Program}(2013)]{hottbook}
{Univalent Foundations Program}, \emph{Homotopy Type Theory: Univalent
  Foundations of Mathematics}.\hskip 1em plus 0.5em minus 0.4em\relax Institute
  for Advanced Study: \url{https://homotopytypetheory.org/book}, 2013.

\bibitem[Vezzosi et~al.(2019)Vezzosi, M\"{o}rtberg, and
  Abel]{vezzosi-mortberg-abel:2019}
\BIBentryALTinterwordspacing
A.~Vezzosi, A.~M\"{o}rtberg, and A.~Abel, ``{Cubical Agda: A Dependently Typed
  Programming Language with Univalence and Higher Inductive Types},'' in
  \emph{Proceedings of the 24th ACM SIGPLAN International Conference on
  Functional Programming}, ser. ICFP '19.\hskip 1em plus 0.5em minus
  0.4em\relax Boston, Massachusetts, USA: ACM, 2019. [Online]. Available:
  \url{http://www.cs.cmu.edu/~amoertbe/papers/cubicalagda.pdf}
\BIBentrySTDinterwordspacing

\bibitem[Vickers(2007)]{vickers:2007}
S.~Vickers, \emph{Locales and Toposes as Spaces}.\hskip 1em plus 0.5em minus
  0.4em\relax Dordrecht: Springer Netherlands, 2007, pp. 429--496.

\bibitem[Yetter(1987)]{yetter:1987}
\BIBentryALTinterwordspacing
D.~Yetter, ``On right adjoints to exponential functors,'' \emph{Journal of Pure
  and Applied Algebra}, vol.~45, no.~3, pp. 287--304, 1987. [Online].
  Available:
  \url{http://www.sciencedirect.com/science/article/pii/0022404987900776}
\BIBentrySTDinterwordspacing

\end{thebibliography}

\ifarXiv
\appendix

\subsection{Normalization structures}

In \cref{fig:norm-sg,fig:norm-pi,fig:norm-s1,fig:norm-s1:ind}, we present the remaining
normalization structures that could not fit in the main body of \cref{sec:stc}.

\begin{figure*}[p]
  \small
  \begin{nf-box}
    \begin{align*}
      \frk{sg} &: \Ext{\IHom{\prn{\Sum{A:\NfTp}{\Prod{x:\Con{var}\prn{A}}{\NfTp}}}}{\NfTp}}{\Syn}{\SynAlg.\Sigma}\\
      \frk{pair} &:
      \Params{A,B}
      \Ext{
        \IHom{\prn{\Sum{x:\Nf{A}}{\Nf{B\prn{x}}}}}{\Nf{\SynAlg.\Sigma\prn{A,B}}}
      }{\Syn}{
        \lambda\prn{a,b}.\prn{a,b}
      }
      \\
      \frk{split} &:
      \Params{\phi,A,B}
      \Ext{
        \IHom{
          \Ne{\phi}{\SynAlg.\Sigma\prn{A,B}}
        }{
          \Sum{x:\Ne{\phi}{A}}{\Ne{\phi}{B\prn{x}}}
        }
      }{\Syn}{
        \lambda\prn{a,b}.\prn{a,b}
      }
    \end{align*}
  \end{nf-box}

  \begin{align*}
    \Sigma &: \Ext{
      \IHom{
        \prn{
          \Sum{A:\Tp}{
            \IHom*{A}{\Tp}
          }
        }\
      }{
        \Tp
      }
    }{\Syn}{\SynAlg.\Sigma}
    \\
    \brk{\Sigma\prn{A,B}} &\cong
    \Sum{x:A}{B\prn{x}}
    \\
    \Con{hcom}\Sub{\Sigma\prn{A,B}}\Sup{r\to s; \phi}p &=
    \Kwd{let}\ x\prn{k} = \Con{hcom}\Sub{A}\Sup{r\to k;\phi}{\lambda i. \brk{i=r\lor_\FF\phi\to \pi_1\prn{p\prn{i}}}}\
    \Kwd{in}\
    \prn{
      x\prn{s},
      \DelimMin{1}
      \Con{com}\Sub{\lambda i.B\prn{x\prn{i}}}\Sup{r\to s; \phi}
      \lambda i.
      \brk{i=r\lor_\FF\phi\to\pi_2\prn{p\prn{i}}}
    }
    \\
    \Con{coe}\Sub{\lambda i.\Sigma\prn{A\prn{i},B\prn{i}}}\Sup{r\to s}p &=
    \Kwd{let}\ x\prn{k} = \Con{coe}\Sub{A}\Sup{r\to k}\pi_1\prn{p}\
    \Kwd{in}\
    \prn{
      x\prn{s},
      \DelimMin{1}
      \Con{coe}\Sub{\lambda i.B\prn{x\prn{i}}}\Sup{r\to s}\pi_2\prn{p}
    }
    \\
    \NormTp{\Sigma\prn{A,B}} &=
    \frk{sg}\prn{
      \NormTp{A},
      \lambda x.
      \NormTp{
        B\prn{
          \Reflect{A}{\bot_\FF}{
            \NeGlue{\frk{var}\prn{x}}{\bot_\FF}{\brk{}\Sub{\SynAlg.\Tm\prn{A}}}
          }
        }
      }
    }
    \\
    \Reflect{\Sigma\prn{A,B}}{\phi}\NeGlue{p_0}{\phi}{p} &=
    \Kwd{let}\ \prn{x_0,y_0} = \frk{split}\prn{p_0};
    \tilde{x} = \Reflect{A}{\phi}\NeGlue{x_0}{\phi}{\pi_1\prn{p}}
    \ \Kwd{in}\
    \prn{
      \tilde{x},
      \Reflect{B\prn{\tilde{x}}}{\phi}{
        \NeGlue{y_0}{\phi}{
          \pi_2\prn{p}
        }
      }
    }
    \\
    \Reify{\Sigma\prn{A,B}}p &=
    \frk{pair}\prn{
      \Reify{A}{\pi_1\prn{p}},
      \Reify{B\prn{\pi_1\prn{p}}}{
        \pi_2\prn{p}
      }
    }
  \end{align*}

  \caption{The cubical normalization structure for dependent sum types.}
  \label{fig:norm-sg}
\end{figure*}

\begin{figure*}[p]
  \small
  \begin{nf-box}
    \begin{align*}
      \frk{pi} &: \Ext{\IHom{\prn{\Sum{A:\NfTp}{\Prod{x:\Con{var}\prn{A}}{\NfTp}}}}{\NfTp}}{\Syn}{\SynAlg.\Pi}\\
      \frk{lam} &: \Params{A,B}
      \Ext{
        \IHom{\prn{\Prod{x:\Con{var}\prn{A}}{\Nf{B\prn{x}}}}}{\Nf{\SynAlg.\Pi\prn{A,B}}}
      }{\Syn}{
        \lambda f.\lambda x.f\prn{x}
      }
      \\
      \frk{app} &: \Params{\phi,A,B}
      \Ext{
        \IHom{\Ne{\phi}{\SynAlg.\Pi\prn{A,B}}}{\Prod{x : \Nf{A}}{\Ne{\phi}{B\prn{x}}}}
      }{\Syn}{\lambda f.\lambda x.f\prn{x}}
    \end{align*}
  \end{nf-box}

  \begin{align*}
    \Pi &: \Ext{
      \IHom{
        \prn{
          \Sum{A:\Tp}{
            \IHom*{A}{\Tp}
          }
        }
      }{\Tp}
    }{\Syn}{\SynAlg.\Pi}
    \\
    \brk{\Pi\prn{A,B}} &\cong
    \Prod{x : A}{B\prn{x}}
    \\
    \Con{hcom}\Sub{\Pi\prn{A,B}}\Sup{r\to s; \phi}f &=
    \lambda x.
    \Con{hcom}\Sub{B}\Sup{r\to s; \phi}
    \lambda i.\brk{i=r\lor_\FF\phi\hookrightarrow f\prn{x,i}}
    \\
    \Con{coe}\Sub{\lambda i.\Pi\prn{A\prn{i},B\prn{i}}}\Sup{r\to s}f &=
    \lambda x.
    \Con{coe}\Sub{
      \lambda i. B\prn{i,
        \Con{coe}\Sub{A}\Sup{s\to i}x
      }
    }\Sup{r\to s}
    f\prn{
      \DelimMin{1}
      \Con{coe}\Sub{A}\Sup{s\to r}x
    }
    \\
    \NormTp{\Pi\prn{A,B}} &=
    \frk{pi}\prn{
      \NormTp{A},
      \lambda x.
      \NormTp{
        B\prn{
          \Reflect{A}{\bot_\FF}{
            \NeGlue{\frk{var}\prn{x}}{\bot_\FF}{\brk{}\Sub{\SynAlg.\Tm\prn{A}}}
          }
        }
      }
    }
    \\
    \Reflect{\Pi\prn{A,B}}{\phi}\NeGlue{f_0}{\phi}{f} &=
    \lambda x.
    \Reflect{B\prn{x}}{\phi}{
      \NeGlue{\frk{app}\prn{f_0, \Reify{A}{x}}}{\phi}{f\prn{x}}
    }
    \\
    \Reify{\Pi\prn{A,B}}f &=
    \frk{lam}\prn{
      \lambda x.
      \Kwd{let}\
      \tilde{x} =
      \Reflect{A}{\bot_\FF}{
        \NeGlue{\frk{var}\prn{x}}{\bot_\FF}{\brk{}\Sub{\SynAlg.\Tm\prn{A}}}
      }
      \
      \Kwd{in}\
      \Reify{B\prn{\tilde{x}}}{
        f\prn{\tilde{x}}
      }
    }
  \end{align*}

  \caption{The cubical normalization structure for dependent product types.}
  \label{fig:norm-pi}
\end{figure*}

\begin{figure*}[p]
  \small
  \begin{nf-box}
    \begin{align*}
      \frk{s}1 &: \Ext{\NfTp}{\Syn}{\SynAlg.\Con{S1}}\\
      \frk{base} &: \Ext{\Nf{\SynAlg.\Con{S1}}}{\Syn}{\SynAlg.\Con{base}}\\
      \frk{loop} &:
      \Ext{
        \Prod{i:\II}{
          \Ext{\Nf{\SynAlg.\Con{S1}}}{\partial{i}}{\frk{base}}
        }
      }{\Syn}{\SynAlg.\Con{loop}}
      \\
      \frk{fhcom} &:
      \Ext{
        \Con{HCom}\prn{\SynAlg.\Con{S1},\Nf{\SynAlg.\Con{S1}}}
      }{\Syn}{
        \SynAlg.\Con{hcom}\Sub{\SynAlg.\Con{S1}}
      }
      \\
      \frk{inds}1 &:
      \Ext{
        \Prod{
          C : \IHom{\Con{var}\prn{\SynAlg.\Con{S1}}}{\NfTp}
        }
        \Prod{
          b : \Nf{C\prn{\SynAlg.\Con{base}}}
        }
        \Prod{
          l : \Prod{i:\II}{
            \Ext{
              \Nf{C\prn{\SynAlg.\Con{loop}\prn{i}}}
            }{\partial{i}}{b}
          }
        }
        \Prod{
          x : \Ne{\phi}{\SynAlg.\Con{S1}}
        }
        \Ne{\phi}{
          C\prn{x}
        }
      }{\Syn}{\SynAlg.\Con{ind}\Sub{\Con{S1}}}
      \\
      \frk{lift} &:
      \Params{\phi}
      \Ext{
        \IHom{
          \prn{
            \Sum{
              b:\Ne{\phi}{\SynAlg.\Con{S1}}
            }{
              \Prod{\_:\brk{\phi}}{
                \Ext{\Nf{\SynAlg.\Con{S1}}}{\Syn}{b}
              }
            }
          }
        }{
          \Con{nf}\prn{\SynAlg.\Con{S1}}
        }
      }{\Syn\lor\phi}{
        \lambda\prn{b,b'}.
        \brk{
          \Syn\hookrightarrow b,
          \phi\hookrightarrow b'
        }
      }
    \end{align*}
  \end{nf-box}

  \bigskip
  \bigskip

  The following definition of $\Con{S1}$ is justified by realignment: the
  constructors give rise to a (quotient-inductive) type that is isomorphic to
  $\SynAlg.\Tm\prn{\SynAlg.\Con{S1}}$ underneath $z:\Syn$; after realignment, the constructor
  $\floors{-}$ becomes the identity function.

  \begin{align*}
    \begin{array}{l}
      \Kwd{data} \ \brk{\Con{S1}} : \Ext{\UU}{\Syn}{\SynAlg.\Tm\prn{\SynAlg.\Con{S1}}}\ \Kwd{where}\\
      \quad
      \begin{array}{l}
        \floors{-} : \Params{\Syn} \to \SynAlg.\Tm\prn{\SynAlg.\Con{S1}}\to\brk{\Con{S1}}\\
        \Con{base} : \Compr{\brk{\Con{S1}}}{\Syn\hookrightarrow\floors{\SynAlg.\Con{base}}}\\
        \Con{loop} : \Prod{i:\II}{
          \Compr{\brk{\Con{S1}}}{
            \partial{i}\hookrightarrow\Con{base},
            \Syn\hookrightarrow\floors{\SynAlg.\Con{loop}\prn{i}}
          }
        }
        \\
        \Con{lift} :
        \Prod{\phi:\FF}
        \Prod{x_0:\Ne{\phi}{\SynAlg.\Con{S1}}}
        \Prod{x:\IHom{\brk{\phi}}{\Ext{\brk{\Con{S1}}}{\Syn}{\floors{x_0}}}}
        \Compr{\brk{\Con{S1}}}{
          \phi\hookrightarrow x,
          \Syn\hookrightarrow \floors{x_0}
        }
        \\
        \Con{fhcom} :
        \Prod{r,s:\II}
        \Prod{\phi:\FF}
        \Prod{c:
          \Prod{i:\II}
          \Prod{\_:\brk{i=r\lor_\FF\phi}}
          \Sum{x:\SynAlg.\Tm\prn{\SynAlg.\Con{S1}}}
          \Ext{\brk{\Con{S1}}}{\Syn}{\floors{x}}
        }
        \Compr{\brk{\Con{S1}}}{
          r=s\lor_\FF\phi\hookrightarrow \pi_2\prn{c\prn{i}},
          \Syn\hookrightarrow
          \SynAlg.\Con{hcom}\Sub{\SynAlg.\Con{S1}}\Sup{r\to s;\phi}\prn{\lambda i.\pi_1\prn{c\prn{i}}}
        }
      \end{array}
    \end{array}
  \end{align*}

  \begin{align*}
    \Con{S1} &: \Ext{\Tp}{\Syn}{\SynAlg.\Con{S1}}\\
    \Con{hcom}\Sub{\Con{S1}}\Sup{r\to s;\phi}{c} &=
    \Con{fhcom}\prn{r,s,\phi,\lambda i.\brk{i=r\lor_\FF\phi\hookrightarrow \prn{c\prn{i},c\prn{i}}}}\\
    \Con{coe}\Sub{\lambda\_.\Con{S1}}\Sup{r\to s}x &= x\\
    \Reflect{\Con{S1}}{\phi}\NeGlue{x_0}{\phi}{x} &= \Con{lift}\prn{\phi,x_0,\brk{\phi\hookrightarrow x}}\\
    \Reify{\Con{S1}}{\floors{c}} &= c\\
    \Reify{\Con{S1}}{\Con{base}} &= \frk{base}\\
    \Reify{\Con{S1}}{\Con{loop}\prn{i}} &= \frk{loop}\prn{i}\\
    \Reify{\Con{S1}}{\Con{lift}\prn{\phi,x_0,x}} &= \frk{lift}_\phi\prn{x_0,\Reify{\Con{S1}}x}\\
    \Reify{\Con{S1}}{\Con{fhcom}\prn{r,s,\phi,c}} &=
    \frk{fhcom}\Sup{r\to s;\phi}\lambda i.\brk{i=r\lor_\FF\phi\hookrightarrow \Reify{\Con{S1}}{\pi_2\prn{c\prn{i}}}}
  \end{align*}

  \caption{Definition of the cubical normalization structure for the circle.}
  \label{fig:norm-s1}
\end{figure*}
\begin{figure*}
  \begin{align*}
    \Con{ind}\Sub{\Con{S1}} &:
    \Ext{
      \Prod{
        C : \IHom{\Con{S1}}{\Tp}
      }
      \Prod{
        b : C\prn{\Con{base}}
      }
      \Prod{
        l : \Prod{i:\II}{
          \Ext{
            C\prn{\Con{loop}\prn{i}}
          }{\partial{i}}{b}
        }
      }
      \Prod{
        x : \Con{S1}
      }
      C\prn{x}
    }{\Syn}{\SynAlg.\Con{ind}\Sub{\Con{S1}}}\\
    \Con{ind}\Sub{\Con{S1}}\prn{C,b,l,\floors{c}} &= \SynAlg.\Con{ind}\Sub{\Con{S1}}\prn{C,b,l,c}\\
    \Con{ind}\Sub{\Con{S1}}\prn{C,b,l,\Con{base}} &= b\\
    \Con{ind}\Sub{\Con{S1}}\prn{C,b,l,\Con{loop}\prn{i}} &= l\prn{i}\\
    \Con{ind}\Sub{\Con{S1}}\prn{C,b,l,\Con{lift}\prn{\phi,x_0,x}} &=
    \!\!\!
    \begin{array}[t]{l}
      \Kwd{let}\ \tilde{x} = \Reflect{\Con{S1}}{\phi}{\NeGlue{x_0}{\phi}{\prn{x_0,x}}}\ \Kwd{in}\\
      \Kwd{let}\ c_0 =
      \frk{inds}1\prn{
        \lambda y.
        \NormTp{
          C\prn{
            \Reflect{\Con{S1}}{\bot_\FF}{\NeGlue{\frk{var}\prn{y}}{\bot_\FF}{\brk{}\Sub{\Tm\prn{\Con{S1}}}}}
          }
        },
        \Reify{C\prn{\Con{base}}}{b},
        \lambda i. \Reify{C\prn{\Con{loop}\prn{i}}}{l\prn{i}},
        x_0
      }
      \ \Kwd{in}\\
      \Reflect{
        C\prn{\tilde{x}}
      }{\phi}{
        \NeGlue{
          c_0
        }{\phi}{
          \Con{ind}\Sub{\Con{S1}}\prn{C,b,l,\tilde{x}}
        }
      }
    \end{array}
    \\
    \Con{ind}\Sub{\Con{S1}}\prn{C,b,l,\Con{fhcom}\prn{r,s,\phi,x}} &=
    \Con{com}\Sub{
      \lambda i.
      C\prn{
        \Con{hcom}\Sub{\Con{S1}}\Sup{r\to i; \phi}x
      }
    }\Sup{r\to s; \phi}
    \lambda i.
    \brk{
      i=r\lor_\FF\phi
      \hookrightarrow
      \Con{ind}\Sub{\Con{S1}}\prn{C,b,l,\pi_2\prn{x\prn{i}}}
    }
  \end{align*}
  \caption{Implementation of the induction principle for the circle.}
  \label{fig:norm-s1:ind}
\end{figure*}

\subsection{Explicit computations}

In \cref{sec:axioms} we gave a presentation of $\CatSTC$
as a category of presheaves, an apparently necessary step to
substantiate the strict universes of synthetic Tait computability; it is still
useful, however, to gain intuitions for the more traditional presentation of
$\CatSTC$ as the comma category $\Comma{\Sh{\RxTop}}{\StrMap^*}$, and to
understand the explicit computations of the inverse image and direct image
parts of the open and closed immersions respectively.

\begin{computation}[Comma category]\label{cmp:comma-cat}
  An object of the comma category $\CatSTC \simeq
  \Comma{\Sh{\RxTop}}{\StrMap^*}$ is an object $E:\Psh{\ThCat}$ together with a
  family of presheaves $\Mor{E'}{\StrMap^*E} : \Psh{\RxCat}$; a morphism from
  $\prn{F,\Mor{F'}{\StrMap^*{F}}}$ to $\prn{E,\Mor{E'}{\StrMap^*{E}}}$ is a
  morphism $\Mor[e]{F}{E} : \ThCat$ together with a commuting square of the
  following kind:
  \[
    \DiagramSquare{
      nw = F',
      sw = \StrMap^*F,
      ne = E,
      se = \StrMap^*E,
      south = \StrMap^*e,
      north = e',
      height = 1.5cm,
    }
  \]
\end{computation}

\begin{computation}[Open immersion]\label{cmp:open-immersion}
  The open immersion $\Mor|open immersion|[\OpImm]{\ThTop}{\GlTop}$
  corresponds under inverse image to the \emph{codomain} fibration
  $\Mor|fibration|[\OpImm^*]{\Comma{\Psh{\RxCat}}{\StrMap^*}}{\Psh{\ThCat}}$.
  Hence we may compute the adjunction $\OpImm^*\dashv\OpImm_*$ as follows:
  \begin{align*}
    \OpImm^* &: \Mor|{|->}|{\prn{E,\Mor{E'}{\StrMap^*E}}}{E}\\
    \OpImm_* &: \Mor|{|->}|{E}{\prn{E,\Mor{\StrMap^*E}{\StrMap^*E}}}
  \end{align*}

  The direct image functor $\OpImm_*$ is fully faithful. In fact, we have two
  additional adjoints $\OpImm_!\dashv\OpImm^*\dashv\OpImm_*\dashv\OpImm^!$, the
  exceptional right adjoint by virtue of the adjunction
  $\StrMap^*\dashv\StrMap_*$.
  \begin{align*}
    \OpImm_! &: \Mor|{|->}|{E}{\prn{E,\Mor{\ObjInit}{\StrMap^*E}}}\\
    \OpImm^! &: \Mor|{|->}|{\prn{E,\Mor{E'}{\StrMap^*E}}}{
      \StrMap_*E'\times_{\StrMap_*\StrMap^*E}{E}
    }
  \end{align*}
\end{computation}

\begin{computation}[Closed immersion]\label{cmp:closed-immersion}
  The closed immersion $\Mor|closed immersion|[\ClImm]{\RxTop}{\GlTop}$
  corresponds under inverse image to the \emph{domain} functor
  $\Mor[\ClImm^*]{\Comma{\Sh{\RxTop}}{\StrMap^*}}{\Psh{\RxCat}}$.
  We compute the adjunction $\ClImm^*\dashv\ClImm_*$ as follows:
  \begin{align*}
    \ClImm^* &: \Mor|{|->}|{\prn{E,\Mor{E'}{\StrMap^*E}}}{E'}\\
    \ClImm_* &: \Mor|{|->}|{E'}{\prn{\ObjTerm,\Mor{E'}{\StrMap^*\ObjTerm}}}
  \end{align*}

  Because $\StrMap$ is an \emph{essential} morphism of topoi, we have an
  additional left adjoint $\ClImm_!\dashv\ClImm^*\dashv\ClImm_*$:
  \begin{align*}
    \ClImm_! &: \Mor|{|->}|{E'}{\prn{\StrMap_!E', \Mor{E'}{\StrMap^*\StrMap_!E'}}}
  \end{align*}
\end{computation}

\begin{computation}[Open and closed modality]
  Using the syntactic open $\Syn$ (\cref{cmp:syntactic-open}), we can compute
  the open modality $\Op = \OpImm_*\OpImm^*$ as the exponential
  $\IHom*{\Syn}{-}$. Likewise, $\Syn$ makes another computation of the closed
  modality $\Cl = \ClImm_*\ClImm^*$ available:
  \[
    \DiagramSquare{
      se/style = pushout,
      nw = A\times\Syn,
      sw = A,
      ne = \Syn,
      se = \Cl{A},
      width = 1.5cm,
      height = 1.5cm,
    }
  \]
\end{computation}

\begin{corollary}
  From
  \cref{cmp:comma-cat,cmp:open-immersion,cmp:syntactic-open,cmp:closed-immersion}
  we may make the following observations:
  \begin{enumerate}
    \item The open modality $\Op := \OpImm_*\OpImm^*$ on $\CatSTC$ has both a right adjoint $\OpImm_*\OpImm^!$ and left adjoint $\OpImm_!\OpImm^*$; hence the syntactic open $\Syn : \CatSTC$ is a tiny object; the adjunction $\OpImm_!\OpImm^*\dashv\Op$ can be computed as $\prn{-\times\Syn}\dashv\IHom*{\Syn}{-}$.
    \item The closed modality $\Cl := \ClImm_*\ClImm^*$ on $\CatSTC$ has a left adjoint $\ClImm_!\ClImm^*$.
    \item The gluing functor $\Mor[\StrMap^*]{\Psh{\ThCat}}{\Psh{\RxCat}}$ can be reconstructed as the composite $\ClImm^*\OpImm_*$.
  \end{enumerate}
\end{corollary}

The following fracture theorem is from SGA~4~\cite{sga:4}.

\begin{lemma}[Fracture~\cite{sga:4,rijke-shulman-spitters:2017}]
  Any sheaf on $\GlTop$ can be reconstructed up to isomorphism from its
  restriction to $\ThTop$ and $\RxTop$; in particular, the following
  square is cartesian for any $A:\CatSTC$:
  \[
    \DiagramSquare{
      nw/style = pullback,
      nw = A,
      ne = \Cl{A},
      se = \Cl\Op{A},
      sw = \Op{A}
    }
  \]
\end{lemma}

\NewDocumentCommand\NfTpAt{m}{\brk{\Tp\ni_{\Con{nf}} #1}}
\NewDocumentCommand\NfAt{mm}{\brk{#1\ni_{\Con{nf}} #2}}
\NewDocumentCommand\NeAt{mmm}{\brk{#3\in^{#1}_{\Con{ne}} #2}}

\subsection{Explicit construction of neutral and normal forms}\label{sec:explicit-nf}

Our construction of the computability model of cubical type theory
(\cref{thm:computability-algebra}) requires only that certain constants
corresponding to the neutral and normal forms exist in $\CatSTC$. However, to
use this computability model to establish the injectivity and (external)
decidability properties of \cref{sec:normalization-result}, it is important to
ensure that the corresponding properties hold for our normal forms.

Concretely, we define them by a family of indexed quotient inductive types
(QITs~\cite{altenkirch-kaposi:2016}) valued in the modal universe $\UUSem$:
\begin{align*}
  \NfTpAt{A} &: \UUSem \qquad \prn{A : \SynAlg.\Tp}\\
  \NfAt{A}{a} &: \UUSem\qquad\prn{A:\SynAlg.\Tp,a : \SynAlg.\Tm\prn{A}}\\
  \NeAt{\phi}{A}{a} &: \UUSem \qquad\prn{\phi:\FF,A:\SynAlg.\Tp,a:\SynAlg.\Tm\prn{A}}
\end{align*}

In fact, we ensure that $\NeAt{\phi}{A}{a}$ is not only $\Op$-connected but
actually $\prn{\phi\lor\Syn}$-connected, capturing the sense in which the data
of a neutral form collapses to a point on its locus of instability.
Then, the collections of normal and neutral forms are obtained by dependent sum
and realignment as follows (noting that the fibers of each family are valued in
$\UUSem$ and are thus $\Op$-connected):
\begin{align*}
  \NfTp &\cong \Sum{A : \SynAlg.\Tp}{\NfTpAt{A}}\\
  \Nf{A} &\cong \Sum{a : \SynAlg.\Tm\prn{A}}{\NfAt{A}{a}}\\
  \Ne{\phi}{A} &\cong \Sum{a : \SynAlg.\Tm\prn{A}}{\NeAt{\phi}{A}{a}}
\end{align*}

Our use of quotienting in the definition of normal forms is to impose correct
\emph{cubical boundaries} on constructors: for instance, we must have
$\partial{i}\to \frk{loop}\prn{i} = \frk{base}$. Because the theory of
cofibrations is (externally) decidable, the quotient can be presented externally
by an effective rewriting system that reduces size and is therefore obviously
noetherian.

\begin{remark}
  An indexed quotient inductive type in $\UUSem$ also has a universal property
  in $\UU$, obtained by adding an additional quotient-inductive clause that
  contracts each fiber to a point under $z:\Syn$.
\end{remark}

In \cref{fig:nf-ne-qit} we present the indexed quotient inductive definition of
normal and neutral forms.

\begin{figure*}
  \begin{mathpar}
    \inferrule[circle type]{
    }{
      \frk{s}{1} : \NfTpAt{\SynAlg.\Con{S1}}
    }
    \and
    \inferrule[path type]{
      \frk{A} : \Prod{i : \II}{\NfTpAt{A\prn{i}}}
      \\
      \frk{a}_0 : \NfAt{A\prn{0}}{a_0}
      \\
      \frk{a}_1 : \NfAt{A\prn{1}}{a_1}\\
    }{
      \frk{path}\Mute{\brc{A,a_0,a_1}}\prn{\frk{A},\frk{a}_0,\frk{a}_1} : \NfTpAt{\SynAlg.\Con{path}\prn{A,a_0,a_1}}
    }
    \and
    \inferrule[glue type]{
      \phi : \FF\\
      \frk{B} : \NfTpAt{B}\\
      \frk{A} : \Prod{z : \brk{\phi}}{\NfTpAt{A\prn{z}}}\\
      \frk{f} : \Prod{z : \brk{\phi}}{\NfAt{\SynAlg.\Con{Equiv}\prn{A\prn{z},B}}{f\prn{z}}}
    }{
      \frk{glue}\Mute{\brc{B,A,f}}\prn{\phi,\frk{B},\frk{A},\frk{f}} : \NfTpAt{\SynAlg.\Con{glue}\prn{\phi,B,A,f}}
    }
    \and
    \inferrule[glue type boundary]{
      \phi : \FF\\
      \frk{B} : \NfTpAt{B}\\
      \frk{A} : \Prod{z : \brk{\phi}}{\NfTpAt{A\prn{z}}}\\
      \frk{f} : \Prod{z : \brk{\phi}}{\NfAt{\SynAlg.\Con{Equiv}\prn{A\prn{z},B}}{f\prn{z}}}\\
      z : \brk{\phi}
    }{
      \frk{glue}\Mute{\brc{B,A,f}}\prn{\phi,\frk{B},\frk{A},\frk{f}} = \frk{A}\prn{z} : \NfTpAt{A\prn{z}}
    }
    \and
    \inferrule[pi/sg type]{
      \frk{A} : \NfTpAt{A}
      \\
      \frk{B} : \Prod{x : \Con{var}\prn{A}} \NfTpAt{B\prn{x}}
    }{
      \frk{pi}\Mute{\brc{A,B}}\prn{\frk{A},\frk{B}} : \NfTpAt{\SynAlg.\Pi\prn{A,B}}
      \\
      \frk{sg}\Mute{\brc{A,B}}\prn{\frk{A},\frk{B}} : \NfTpAt{\SynAlg.\Sigma\prn{A,B}}
    }
  \end{mathpar}

  \begin{mathpar}
    \inferrule[unstable]{
      z : \brk{\phi}
    }{
      \star\prn{z} : \NeAt{\phi}{A}{a}
    }
    \and
    \inferrule[unstable collapse]{
      z : \brk{\phi}\\
      \frk{a} : \NeAt{\phi}{A}{a}
    }{
      \frk{a} = \star\prn{z} : \NeAt{\phi}{A}{a}
    }
    \and
    \inferrule[variable]{
      x : \Con{var}\prn{A}
    }{
      \frk{var}\Mute{\brc{A}}\prn{x} : \NeAt{\bot}{A}{x}
    }
    \and
    \inferrule[function application]{
      \frk{f} : \NeAt{\phi}{\SynAlg.\Pi\prn{A,B}}{f}\\
      \frk{a} : \NfAt{A}{a}
    }{
      \frk{app}\Mute{\brc{A,B,f,a}}\prn{\frk{f},\frk{a}} : \NeAt{\phi}{A}{a}
    }
    \and
    \inferrule[pair projection]{
      \frk{p} : \NeAt{\phi}{\SynAlg.\Sigma\prn{A,B}}{p}
    }{
      \frk{fst}\Mute{\brc{A,B,p}}\prn{\frk{p}} : \NeAt{\phi}{A}{\pi_1\prn{p}}\\
      \frk{snd}\Mute{\brc{A,B,p}}\prn{\frk{p}} : \NeAt{\phi}{B\prn{\pi_1\prn{p}}}{\pi_2\prn{p}}
    }
    \and
    \inferrule[path application]{
      \frk{p} : \NeAt{\phi}{\SynAlg.\Con{path}\prn{A,a_0,a_1}}{p}
      \\
      r : \II
    }{
      \frk{papp}\Mute{\brc{A,a_0,a_1}}\prn{\frk{p},r} : \NeAt{\phi\lor_\FF\partial{r}}{A\prn{r}}{p\prn{r}}
    }
    \and
    \inferrule[unglue destructor]{
      \frk{g} : \NeAt{\psi}{\SynAlg.\Con{glue}\prn{\phi,B,A,f}}{g}
    }{
      \frk{unglue}\Mute{\brc{B,A,f,g}}\prn{\phi,\frk{g}} : \NeAt{\psi\lor_\FF\phi}{B}{\SynAlg.\Con{unglue}\prn{g}}
    }
    \and
    \inferrule[circle induction]{
      \frk{C} : \Prod{x : \Con{var}\prn{\SynAlg.\Con{S1}}}{\NfTpAt{C\prn{x}}}\\
      \frk{b} :\NfAt{C\prn{\SynAlg.\Con{base}}}{b}\\
      \frk{l} : \Prod{i:\II}{\NfAt{C\prn{\SynAlg.\Con{loop}\prn{i}}}{l\prn{i}}}\\
      \frk{s} : \NeAt{\phi}{\SynAlg.\Con{S1}}{s}
    }{
      \frk{ind}\Mute{\brc{C,b,l,s}}\prn{\frk{C},\frk{b},\frk{l},\frk{s}} : \NeAt{\phi}{C\prn{s}}{\SynAlg.\Con{ind}\Sub{\Con{S1}}\prn{C,b,l,s}}
    }
  \end{mathpar}

  \begin{mathpar}
    \inferrule[circle neutral lift]{
      \frk{s}_0 : \NeAt{\phi}{\SynAlg.\Con{S1}}{s}\\
      \frk{s}_\phi : \Prod{\_:\brk{\phi}}{\NfAt{\SynAlg.\Con{S1}}{s}}
    }{
      \frk{lift}\Mute{\brc{s}}\prn{\frk{s}_0,\frk{s}_\phi} : \NfAt{\SynAlg.\Con{S1}}{s}
    }
    \and
    \inferrule[circle neutral lift boundary]{
      \frk{s}_0 : \NeAt{\phi}{\SynAlg.\Con{S1}}{s}\\
      \frk{s}_\phi : \Prod{\_:\brk{\phi}}{\NfAt{\SynAlg.\Con{S1}}{s}}\\
      z : \brk{\phi}
    }{
      \frk{lift}\Mute{\brc{s}}\prn{\frk{s}_0,\frk{s}_\phi} = \frk{s}_\phi\prn{z} : \NfAt{\SynAlg.\Con{S1}}{s}
    }
    \and
    \inferrule[circle base]{
    }{
      \frk{base} : \NfAt{\SynAlg.\Con{S1}}{\SynAlg.\Con{base}}
    }
    \and
    \inferrule[circle loop]{
      r : \II
    }{
      \frk{loop}\prn{r} : \NfAt{\SynAlg.\Con{S1}}{\SynAlg.\Con{loop}\prn{r}}
    }
    \and
    \inferrule[circle loop boundary]{
      r : \II
      \\
      \_ : \brk{\partial r}
    }{
      \frk{loop}\prn{r} = \frk{base} : \NfAt{\SynAlg.\Con{S1}}{\SynAlg.\Con{base}}
    }
    \and
    \inferrule[circle formal homogeneous composition]{
      r,s:\II\\
      \phi:\FF\\
      \frk{a} : \Prod{i:\II}\Prod{z:\brk{i=r\lor_\FF\phi}}\NfAt{\SynAlg.\Con{S1}}{a\prn{i,z}}
    }{
      \frk{fhcom}\Mute{\brc{a}}\prn{r,s,\phi,\frk{a}} : \NfAt{\SynAlg.\Con{S1}}{
        \SynAlg.\Con{hcom}\prn{\SynAlg.\Con{S1},r,s,\phi,a}
      }
    }
    \and
    \inferrule[function abstraction]{
      \frk{f} : \Prod{x : \Con{var}\prn{A}}{\NfAt{B\prn{x}}{f\prn{x}}}
    }{
      \frk{lam}\Mute{\brc{A,B,f}}\prn{\frk{f}} : \NfAt{\SynAlg.\Pi\prn{A,B}}{\lambda x.f\prn{x}}
    }
    \and
    \inferrule[pair constructor]{
      \frk{a} : \NfAt{A}{a}\\
      \frk{b} : \NfAt{B\prn{a}}{b}
    }{
      \frk{pair}\Mute{\brc{A,B,a,b}}\prn{\frk{a},\frk{b}} : \NfAt{\SynAlg.\Sigma}{\prn{a,b}}
    }
    \and
    \inferrule[path abstraction]{
      \frk{p} : \Prod{i: \II}{\NfAt{A\prn{i}}{p\prn{i}}}
    }{
      \frk{plam}\Mute{\brc{A,a_0,a_1,p}}\prn{\frk{p}} : \NfAt{\SynAlg.\Con{path}\prn{A,a_0,a_1}}{\lambda i.p\prn{i}}
    }
    \and
    \inferrule[englue constructor]{
      \frk{a} : \Prod{z : \brk{\phi}}{\NfAt{A}{a}}\\
      \frk{b} : \NfAt{B}{b}\\
      \forall z:\brk{\phi}. \frk{b} = \frk{a}\prn{z}
    }{
      \frk{englue}\Mute{\brc{B,A,f,a,b}}\prn{\phi,\frk{a},\frk{b}} : \NfAt{
        \SynAlg.\Con{glue}\prn{\phi,B,A,f}
      }{\SynAlg.\Con{glue}/\Con{tm}\prn{a,b}}
    }
    \and
    \inferrule[englue constructor boundary]{
      \frk{a} : \Prod{z : \brk{\phi}}{\NfAt{A}{a}}\\
      \frk{b} : \NfAt{B}{b}\\
      \forall z:\brk{\phi}. \frk{b} = \frk{a}\prn{z}\\
      z : \brk{\phi}
    }{
      \frk{englue}\Mute{\brc{B,A,f,a,b}}\prn{\phi,\frk{a},\frk{b}} = \frk{a}\prn{z} : \NfAt{A}{a}
    }
    \and
  \end{mathpar}

  \caption{The explicit indexed quotient-inductive definition of normal and neutral forms. The \textsc{unstable} and \textsc{unstable collapse} rules ensure that $\Ne{\phi}{A}$ collapses to $\SynAlg.\Tm\prn{A}$ within the locus of instability $\phi$.}
  \label{fig:nf-ne-qit}
\end{figure*}

\begin{lemma}[Decidability of equality of normal forms]\label{thm:dec-eq-nf}
  Given two external normal forms $\Mor[A_0,A_1]{\Atom{\Gamma}}{\NfTp}$,
  it is recursively decidable whether $A_0=A_1$ or $A_0\not=A_1$.
\end{lemma}

\begin{proof}
  As in the normal form presentation of strict coproducts
  \citep{altenkirch-dybjer-hofmann-scott:2001}, elements of $\NfTp$ are not pure
  data: they include binders of type $\II$, $\Con{var}\prn{A}$, and
  $\brk{\phi}$. Nevertheless, equality is algorithmically decidable as follows,
  by recursion on $\Gamma,A_0,A_1$.

  At a binder of type $\II$ or $\Con{var}\prn{A}$, we continue at $\Gamma.\II$
  or $\Gamma.A$ respectively. At a binder of type $\brk{\phi}$, we note that our
  definition of $\FF$ (\cref{sec:cof}) licenses a case split on the form of
  $\phi:\FF$. We eliminate universal quantifications in the style of
  \citet{cchm:2017}, proceeding by ``left inversion'' until reaching a
  conjunction of equations $\overline{r=_{\II}s}$. If the conjunction implies
  $0=1$, we halt; otherwise, we proceed under the equalizing atomic substitution
  $\Delta\to\Gamma$.
\end{proof}

\begin{lemma}[Injectivity of normal form constructors]\label{lem:nf-con-inj}
  The following formula holds in the internal logic of $\CatSTC = \Sh{\GlTop}$:
  \begin{multline*}
    \forall A,A',B,B'.\\
    \frk{pi}\prn{A,B} = \frk{pi}\prn{A',B'}
    \implies
    \Cl\prn{
      \prn{A,B} = \prn{A',B'}
    }
  \end{multline*}
\end{lemma}

\begin{proof}
  We prove this in the same way that one proves injectivity of constructors of
  any inductive type, with one subtlety: we must ensure that our constructions
  respect the cubical boundary of $\frk{glue}$, which is the only constructor of
  $\NfTpAt{-}$ subject to an equational clause.

  By induction on $\NfTpAt{-}$, we define a $\Op$-connected predicate
  $\Con{isPi} : \IHom{\NfTp}{\Omega_{\Cl}}$ satisfying a universal property:
  \[
    \forall X:\NfTp.
      \Con{isPi}\prn{X} \iff
      \Cl\prn{\exists\frk{A},\frk{B}. X = \frk{pi}\prn{\frk{A},\frk{B}}}
  \]

  We define $\Con{isPi}$ as follows:
  \begin{align*}
  \Con{isPi}\prn{\_,\frk{pi}\prn{\frk{A},\frk{B}}} &= \Cl\top\\
  \Con{isPi}\prn{\_,\frk{glue}\Mute{\brc{B,A,f}}\prn{\phi,\frk{B},\frk{A},\frk{f}}} &=
    \Cl\exists{z:\brk{\phi}}.{\Con{isPi}\prn{\Mute{A},\frk{A}\prn{z}}}\\
  \Con{isPi}\prn{\_,\frk{sg}\prn{\frk{A},\frk{B}}} &= \Cl\bot\\
  \Con{isPi}\prn{\_,\frk{path}\prn{\frk{A},\frk{a}_0,\frk{a}_1}} &= \Cl\bot\\
  \Con{isPi}\prn{\_,\frk{s}1} &= \Cl\bot
  \end{align*}
  We must verify that $\Con{isPi}$ respects the $\phi$-boundary of $\frk{glue} =
  \frk{A}(z)$. But supposing that $z:\brk{\phi}$, we can calculate that
  $\Con{isPi}\prn{\_,\frk{glue}\prn{\phi,\frk{B},\frk{A},\frk{f}}} =
  \Con{isPi}\prn{\_,\frk{A}(z)}$, which is exactly what it means for
  $\Con{isPi}$ to respect that equation.

  The reverse implication of the universal property of $\Con{isPi}$ is
  immediate. We prove the forward implication by induction on $X$, where again
  the only subtlety is in the case of $\frk{glue}$: if
  $\Con{isPi}\prn{\_,\frk{glue}\prn{\phi,\frk{B},\frk{A},\frk{f}}}$ then we must
  have $\Con{isPi}\prn{\Mute{A},\frk{A}\prn{\Mute{z}}}$, which by the inductive
  hypothesis implies $\frk{A}\prn{\Mute{z}} = \frk{pi}$.

  Then we define $\Con{dom} : \IHom{\Compr{X:\NfTp}{\Con{isPi}\prn{X}}}{\Cl\NfTp}$:
  \begin{align*}
  \Con{dom}\prn{\_,\frk{pi}\prn{\frk{A},\frk{B}}} &= \ClRet{\frk{A}}\\
  \Con{dom}\prn{\_,\frk{glue}\prn{\phi,\frk{B},\frk{A},\frk{f}}} &=
    \Con{dom}\prn{\Mute{A},\frk{A}}
  \end{align*}
  Note that these are the only constructors that may satisfy $\Con{isPi}$, and
  in the case of $\frk{glue}$, if
  $\Con{isPi}\prn{\_,\frk{glue}\prn{\phi,\frk{B},\frk{A},\frk{f}}}$ holds then $\phi=\top$ and
  $\Con{isPi}\prn{\Mute{A},\frk{A}}$. We may define $\Con{cod}$ similarly.

  Finally, suppose that $\frk{pi}\prn{A,B} = \frk{pi}\prn{A',B'}$. Then by
  applying $\Con{dom}$ to both sides of this equation, we get $\Cl\prn{A = A'}$;
  the equality of codomains likewise follows by applying $\Con{cod}$.
\end{proof}

\subsection{Proofs of theorems}

\LemIntervalPreserved*
\begin{proof}
  First we compute the representable points of $\StrMap^*\Yo[\ThCat]{\II}$ by transpose and the Yoneda lemma:
  \begin{align*}
    &\Hom[\Psh{\RxCat}]{\Yo[\RxCat]{\Gamma}}{\StrMap^*\Yo[\ThCat]{\II}}
    \\
    &\quad\cong
    \Hom[\Psh{\ThCat}]{\StrMap_!\Yo[\RxCat]{\Gamma}}{\Yo[\ThCat]{\II}}
    \\
    &\quad\cong
    \Hom[\Psh{\ThCat}]{\Yo[\ThCat]{\StrMap{\Gamma}}}{\Yo[\ThCat]{\II}}
    \\
    &\quad\cong
    \Hom[\ThCat]{\StrMap{\Gamma}}{\II}
  \end{align*}

  We see by induction on the definition of the objects and hom sets of
  $\RxCat$ that this is equivalent to $\Hom[\RxCat]{\Gamma}{\cdot.\II}$.
\end{proof}

\LemGluingTinyObjects*
\begin{proof}
  We must check that the exponential functor $\IHom*{X}{-}$ preserves colimits.
  Fixing a diagram $\Mor[E_\bullet]{I}{\Sh{\XTop}}$, we may compute the
  exponential $\IHom*{X}{\Colim{I}{E_\bullet}}$ in the language of $\Sh{\YTop}$
  as follows; first, the standard computation that glues a function from the
  open subtopos onto a function from the closed subtopos~\cite{johnstone:2002}:

  \vspace{-1em}
  \begingroup
  \small
  \[
    \DiagramSquare{
      width = 4cm,
      nw/style = pullback,
      ne = \IHom{i^*X}{i^*\Colim{I}{E_\bullet}},
      se = \IHom{i^*X}{f^*j^*\Colim{I}{E_\bullet}},
      sw = f^*\IHom*{j^*X}{j^*\Colim{I}{E_\bullet}},
      nw = i^*\IHom*{X}{\Colim{I}{E_\bullet}},
      west/style = exists,
    }
  \]
  \endgroup

  Commute cocontinuous functors past colimits.

  \vspace{-1em}
  \begingroup
  \small
  \[
    \DiagramSquare{
      width = 4cm,
      nw/style = pullback,
      ne = \IHom{i^*X}{\Colim{I}{i^* E_\bullet}},
      se = \IHom{i^*X}{\Colim{I}{f^*j^* E_\bullet}},
      sw = f^*\IHom*{j^*X}{\Colim{I}{j^* E_\bullet}},
      nw = i^*\IHom*{X}{\Colim{I}{E_\bullet}},
      west/style = exists,
    }
  \]
  \endgroup

  Use the tininess of $i^*X,j^*X$ and the cocontinuity of $f^*$.

  \vspace{-1em}
  \begingroup
  \small
  \[
    \DiagramSquare{
      width = 4cm,
      nw/style = pullback,
      ne = {\Colim{I}{\IHom*{i^*X}{i^* E_\bullet}}},
      se = {\Colim{I}{\IHom*{i^*X}{f^*j^* E_\bullet}}},
      sw = {\Colim{I}{f^*\IHom*{j^*X}{j^* E_\bullet}}},
      nw = i^*\IHom*{X}{\Colim{I}{E_\bullet}},
      west/style = exists,
    }
  \]
  \endgroup

  Hence by the universality of colimits we have:

  \vspace{-1em}
  \begingroup
  \small
  \[
    \DiagramSquare{
      width = 4cm,
      nw/style = pullback,
      ne = {\Colim{I}{\IHom*{i^*X}{i^* E_\bullet}}},
      se = {\Colim{I}{\IHom*{i^*X}{f^*j^* E_\bullet}}},
      sw = {\Colim{I}{f^*\IHom*{j^*X}{j^* E_\bullet}}},
      nw = {\Colim{I}{i^*\IHom*{X}{E_\bullet}}},
      west/style = exists,
    }
    \qedhere
  \]
  \endgroup
\end{proof}

\begin{lemma}\label{lem:formal-cxs-not-false}
  If $\Gamma:\RxCat$ is a formal context, then the hom set $\Hom[\ThCat]{\StrMap{\Gamma}}{\brk{0=1}}$ is empty.
\end{lemma}

\begin{proof}
  Formal contexts do not induce assumptions of false cofibrations. This can be
  seen by a model construction in which $\ThCat$ is interpreted into cubical
  sets, where the interpretation of cofibrations is standard but each type is
  interpreted as an \emph{inhabited} type. Such an argument accommodates types
  that are ``weakly empty'' (e.g.\ the $\Con{void}$ type lacking an
  $\eta$-law), because the $\Con{abort}$ eliminator can simply return the
  basepoint of its motive.
\end{proof}

\BotImpliesSyn*

\begin{proof}
  We recall that the interval $\II$ is purely syntactic
  (\cref{lem:ii-preserved}), hence we have $\FakeFalse =
  \OpImm_*\Yo[\ThCat]{\brk{0=1}}$ and therefore $\ClImm^*\FakeFalse =
  \StrMap^*\Yo[\ThCat]{\brk{0=1}}$. To show the inequality $\FakeFalse\leq\Syn$, we
  need to exhibit a square of the following kind in $\Psh{\RxCat}$:
  \[
    \DiagramSquare{
      nw = \StrMap^*\Yo[\ThCat]{\brk{0=1}},
      sw = \StrMap^*\Yo[\ThCat]{\brk{0=1}},
      ne = \ObjInit,
      se = \ObjTerm,
      east = \Syn,
      west = \FakeFalse,
      north/style = exists,
    }
  \]

  It suffices to show that $\StrMap^*\Yo[\ThCat]{\brk{0=1}}$ is the initial
  object of $\Psh{\RxCat}$, but this follows from \cref{lem:formal-cxs-not-false}.
\end{proof}

\begin{lemma}\label{lem:01-connectedness}
  A type $A : \UU$ in $\CatSTC$ is $\FakeFalse$-connected if and only $\Op{A}$ is $\FakeFalse$-connected.
\end{lemma}

\begin{proof}
  If $A$ is $\FakeFalse$-connected, it is immediate that $\Op{A}$ is
  $\FakeFalse$-connected. Conversely, assume that $\Op{A}$ is
  $\FakeFalse$-connected; \cref{lem:0=1-implies-syn} implies that $A$ is also
  $\FakeFalse$-connected.
\end{proof}

\LemLocality*
\begin{proof}
  For $\FakeFalse$-connectedness, we use \cref{lem:01-connectedness}. Then,
  fixing $\phi,\psi : \FF$ we must check that a partial element
  $\Mor[\brk{\phi\hookrightarrow a_\phi,\psi\hookrightarrow a_\psi}]{\brk{\phi}\lor\brk{\psi}}{A}$ can be extended to a unique partial
  element $\Mor{\brk{\phi\lor_\FF\psi}}{A}$.

  We assume a proof of $\brk{\phi\lor_\FF\psi}$; by \cref{con:disj} we have
  $\SynAlg.\brk{\phi\mathbin{\SynAlg.\lor_\FF}\psi}$ and
  $\Cl\prn{\brk{\phi}\lor\brk{\psi}}$, which is the same as
  $\brk{\phi}\lor\brk{\psi}\lor\Syn$. Hence we may form the following
  partial element, using the fact that $\Op{A}$ is $\FF$-local:
  \[
    \left\lbrack
    \begin{array}{ll}
      \phi\hookrightarrow a_\phi\\
      \psi\hookrightarrow a_\psi\\
      \Syn\hookrightarrow \brk{\phi\hookrightarrow a_\phi,\psi\hookrightarrow a_\psi}\Sub{\Op{A}}
    \end{array}
    \right .
    \qedhere
  \]
\end{proof}

\begin{lemma}\label{lem:atomic-nerve}
  For any $X:\CatSTC$, we have a canonical isomorphism $\Hom*{\Atom}{X} \cong
  \ClImm^*X : \Psh{\RxCat}$ determined by adjoint transpose and the Yoneda
  lemma.
\end{lemma}

\begin{proof}
  Fix an atomic context $\Gamma : \RxCat$ and compute:
  \begin{align*}
    \Hom*{\Atom}{X}\prn{\Gamma} &\cong \Hom[\CatSTC]{\Atom{\Gamma}}{X}\\
    &= \Hom[\CatSTC]{\ClImm_!\Yo[\RxCat]{\Gamma}}{X}\\
    &\cong \Hom[\Psh{\RxCat}]{\Yo[\RxCat]{\Gamma}}{\ClImm^*X}\\
    &\cong \ClImm^*X\prn{\Gamma}
    \qedhere
  \end{align*}
\end{proof}

\ReixAtom*
\begin{proof}
  The vertical map is computed in the language of the comma category as the
  following square:
  \[
    \DiagramSquare{
      width = 3.5cm,
      nw = \Hom*{\bbrk{-}}{X},
      ne = \Hom*{\Atom}{X} \cong \ClImm^*X,
      sw = \StrMap^*\OpImm^*X,
      se = \StrMap^*\OpImm^*X,
      east = X,
      south = \StrMap^*\ArrId{\OpImm^*X},
      north = \Hom*{\AtmToCx}{X},
      west = X\Sub{\CmpAlg}
    }
  \]

  To see that the diagram commutes, we chase an element
  $\Mor[x]{\bbrk{\Gamma}}{X}$, using the fact that each component
  $\Mor[\AtmToCx_\Gamma]{\Atom{\Gamma}}{\bbrk{\Gamma}}$ is vertical.
\end{proof}

\NormFunction*
\begin{proof}
  Unfolding things more precisely,
  the vertical map $\Mor{\SynAlg.\Tp}{\prn{\CmpAlg.\Tp}\Sub{\CmpAlg}}$
  must be a square of the following form:
  \[
    \DiagramSquare{
      width = 3.5cm,
      nw = \StrMap^*\Yo[\ThCat]{\Tp},
      sw = \StrMap^*\Yo[\ThCat]{\Tp},
      se = \StrMap^*\Yo[\ThCat]{\Tp},
      south = \StrMap^*\ArrId{\Yo[\ThCat]{\Tp}},
      west = {\SynAlg.\Tp = \OpImm_*\Yo[\ThCat]{\Tp}},
      ne = \Hom*{\bbrk{-}}{\CmpAlg.\Tp},
      east = \prn{\CmpAlg.\Tp}\Sub{\CmpAlg},
      north/style = exists,
    }
  \]

  The upstairs map is defined by functoriality of the computability
  interpretation, taking a type $\Mor[A]{\StrMap{\Gamma}}{\Tp}$ to its chosen
  normalization structure $\Mor[\CmpAlg\prn{A}]{\bbrk{\Gamma}}{\CmpAlg.\Tp}$.
\end{proof}

\Injectivity*
\begin{proof}
  By completeness, $\SynAlg.\Pi\prn{A,B} = \SynAlg.\Pi\prn{A',B'}$ implies:
  \[
    \frk{pi}\prn{\Con{nbe}\prn{A},\lambda x.\Con{nbe}\prn{B\prn{x}}} =
    \frk{pi}\prn{\Con{nbe}\prn{A'},\lambda x.\Con{nbe}\prn{B'\prn{x}}}
  \]
  By soundness, it suffices to show:
  \[
    \Cl\prn{\prn{\Con{nbe}\prn{A},\lambda x.\Con{nbe}\prn{B\prn{x}}} =
    \prn{\Con{nbe}\prn{A'},\lambda x.\Con{nbe}\prn{B'\prn{x}}}}
  \]
  The result follows from injectivity of $\frk{pi}$ in $\NfTp$, shown by a
  standard inductive argument (\cref{lem:nf-con-inj}).
\end{proof}

\begin{notation}\label{not:cmp-maps}
  We write $\bbrk{-} : \SynAlg.\Tp \to \CmpAlg.\Tp$ and $\bbrk{-} :
  \SynAlg.\Tm\prn{A} \to \bbrk{A}$ for the maps constructed in
  \cref{thm:norm-fn} which send types and terms to their normalization data in
  $\CmpAlg$.
\end{notation}

\begin{lemma}[Idempotence for variables]\label{lem:stab-var}
  For all $x:\Con{var}\prn{A}$, we have $\bbrk{x} = \Reflect{\bbrk{A}}{\bot}{\NeGlue{\frk{var}\prn{x}}{\bot}{\brk{}}}$.
\end{lemma}

\begin{proof}
  Because this equation already holds when restricted over the open subtopos,
  it suffices to reason ``upstairs'' in $\Sh{\RxTop}$, hence pointwise with
  respect to an arbitrary atomic context $\Gamma:\RxCat$.  Here $x$ is an atomic
  term $\IsVar{\Gamma}{x}{A}$ and $\bbrk{x}$ is a function that projects the corresponding cell
  from any atomic substitution $\gamma : \Yo[\RxCat]{\Gamma}$ and
  \emph{reflects it} in the chosen normalization structure $\bbrk{A}\gamma$ as
  $\Reflect{\bbrk{A}\gamma}{\bot}{\NeGlue{\frk{var}\prn{x\gamma}}{\bot}{\brk{}}}$,
  recalling \cref{con:atm-to-cx}.
\end{proof}

\begin{lemma}[Idempotence of normalization]\label{lem:norm-stab}
  The maps of \cref{not:cmp-maps} satisfy the following equations:
  \begin{enumerate}
    \item For all $x : \Ne{\phi}{A}$, we have
    $\bbrk{x} = \Reflect{\bbrk{A}}{\phi}{\NeGlue{x}{\phi}{\bbrk{x}}}$.

    \item For all $a : \Nf{A}$, we have $\Reify{\bbrk{A}}{\bbrk{a}} = a$.

    \item For all $A:\NfTp$, we have $\NormTp{\bbrk{A}} = A$.
  \end{enumerate}

  Unfolding the definition of the normalization function, we therefore have
  $\Con{nbe}\prn{A} = A$ for any normal type $A:\NfTp$, and
  $\Con{nbe}\prn{a} = a : \Nf{A}$ for any normal term $a:\Nf{A}$.
\end{lemma}

\begin{proof}
  By simultaneous induction on the normal and neutral forms, using essentially
  the argument of \citet{kaposi:thesis}. Because the syntactic part of the
  normalization function is the identity, it suffices to reason in the language
  of $\Sh{\RxTop}$. Representative cases follow:

  \begin{enumerate}

    \item The case for variables is \cref{lem:stab-var}.

    \item The case for neutral function application $\frk{app}\prn{f,a}$ is
      as follows. By induction, $\bbrk{f} =
      \Reflect{\bbrk{\SynAlg.\Pi\prn{A,B}}}{\phi}{\NeGlue{f}{\phi}{\bbrk{f}}}$ and
      $\Reify{\bbrk{A}}{\bbrk{a}} = a$, and we need to check that
      $\bbrk{\frk{app}\prn{f,a}}$ is equal to
      $\Reflect{\bbrk{B\prn{a}}}{\phi}{\NeGlue{\frk{app}\prn{f,a}}{\phi}{\bbrk{\frk{app}\prn{f,a}}}}$.
      \begin{align*}
        &\bbrk{\frk{app}\prn{f,a}} \\
        &\quad=
        \bbrk{f}\bbrk{a}
        \\
        &\quad=
        \prn{
          \Reflect{\bbrk{\SynAlg.\Pi\prn{A,B}}}{\phi}{\NeGlue{f}{\phi}{\bbrk{f}}}
        }\bbrk{a}
        \\
        &\quad=
        \prn{
          \Reflect{\Pi\prn{\bbrk{A},\lambda x.\bbrk{B\prn{x}}}}{\phi}{\NeGlue{f}{\phi}{\bbrk{f}}}
        }\bbrk{a}
        \\
        &\quad=
        \Reflect{\bbrk{B\prn{a}}}{\phi}{
          \NeGlue{
            \frk{app}\prn{f,\Reify{\bbrk{A}}{\bbrk{a}}}
          }{\phi}{
            \bbrk{f}\bbrk{a}
          }
        }
        \\
        &\quad=
        \Reflect{\bbrk{B\prn{a}}}{\phi}{
          \NeGlue{
            \frk{app}\prn{f,a}
          }{\phi}{
            \bbrk{f}\bbrk{a}
          }
        }
        \\
        &\quad=
        \Reflect{\bbrk{B\prn{a}}}{\phi}{
          \NeGlue{
            \frk{app}\prn{f,a}
          }{\phi}{
            \bbrk{\frk{app}\prn{f,a}}
          }
        }
      \end{align*}

    \item The case for neutral path application $\frk{papp}\prn{p,r}$ follows.
      \begingroup\scriptsize
      \begin{align*}
        &\bbrk{\frk{papp}\prn{p,r}}
        \\
        &\quad= \bbrk{p}\prn{r}
        \\
        &\quad=
        \prn{\Reflect{\bbrk{\SynAlg.\Con{path}\prn{A,a_0,a_1}}}{\phi}\NeGlue{p}{\phi}{\bbrk{p}}}
        \prn{r}
        \\
        &\quad=
        \Reflect{\bbrk{A\prn{r}}}{\phi\lor_\FF\partial{r}}{
          \NeGlue{\frk{papp}\prn{p,r}}{\phi\lor_\FF\partial{r}}{
            \brk{
              \phi\hookrightarrow \bbrk{p}\prn{r},
              \overline{r=\epsilon\hookrightarrow a_\epsilon}
            }
          }
        }
        \\
        &\quad=
        \Reflect{\bbrk{A\prn{r}}}{\phi\lor_\FF\partial{r}}{
          \NeGlue{\frk{papp}\prn{p,r}}{\phi\lor_\FF\partial{r}}{
            \bbrk{p}\prn{r}
          }
        }
        \\
        &\quad=
        \Reflect{\bbrk{A\prn{r}}}{\phi\lor_\FF\partial{r}}{
          \NeGlue{\frk{papp}\prn{p,r}}{\phi\lor_\FF\partial{r}}{
            \bbrk{\frk{papp}\prn{p,r}}
          }
        }
      \end{align*}
      \endgroup

    \item The case for stabilizing a neutral element of the circle is as
      follows. Starting with a neutral $s_0$ such that $\bbrk{s_0} =
      \Reflect{\Con{S1}}{\phi}{\NeGlue{s_0}{\phi}{\bbrk{s_0}}}$ and a partial
      normal form $s_\phi$ such that $\Reify{\Con{S1}}{\bbrk{s_\phi}} =
      s_\phi$, we compute:
      \begin{align*}
        &\Reify{\Con{S1}}{\bbrk{\frk{lift}_\phi\prn{s_0,s_\phi}}}
        \\
        &\quad=
        \Reify{\Con{S1}}{\bbrk{s_0}}
        \\
        &\quad=
        \Reify{\Con{S1}}{\Reflect{\Con{S1}}{\phi}{\NeGlue{s_0}{\phi}{\bbrk{s_0}}}}
        \\
        &\quad=
        \Reify{\Con{S1}}{\Reflect{\Con{S1}}{\phi}{\NeGlue{s_0}{\phi}{\bbrk{s_\phi}}}}
        \\
        &\quad=
        \Reify{\Con{S1}}{
          \Con{lift}\prn{\phi,s_0,\bbrk{s_\phi}}
        }
        \\
        &\quad=
        \frk{lift}_\phi\prn{s_0,\Reify{\Con{S1}}{\bbrk{s_\phi}}}
        \\
        &\quad=
        \frk{lift}_\phi\prn{s_0,s_\phi}
      \end{align*}

    \item The case for the dependent product type constant is as follows.
      \begin{align*}
        &\NormTp{\bbrk{\frk{pi}\prn{A,B}}}
        \\
        &\quad= \NormTp{
          \Pi\prn{
            \bbrk{A},
            \lambda x.
            \bbrk{B\prn{x}}
          }
        }
        \\
        &\quad=
        \frk{pi}\prn{
          \NormTp{\bbrk{A}},
          \lambda x.\NormTp{\bbrk{B\prn{x}}}
        }
        \\
        &\quad=
        \frk{pi}\prn{A, B}
        \qedhere
      \end{align*}
  \end{enumerate}
\end{proof}

 \fi

\end{document}